%% file: main.tex
\pdfoutput=1
\documentclass[11pt]{article}
\usepackage{fullpage}
\usepackage{amsmath,amsfonts,amsthm,amssymb,multirow}
\usepackage{times}
\usepackage{graphicx}
\usepackage{floatpag}
\usepackage{algorithm}
\usepackage[noend]{algpseudocode}
\usepackage{bbold}
\usepackage{enumitem}
\usepackage{subcaption} 
\usepackage{calc}
\usepackage{tikz}
\usetikzlibrary{decorations.markings}
\tikzstyle{vertex}=[circle, draw, inner sep=0pt, minimum size=6pt]

\usepackage{graphicx}
\usepackage{tabularx}
\makeatletter
\newcommand{\multiline}[1]{%
  \begin{tabularx}{\dimexpr\linewidth-\ALG@thistlm}[t]{@{}X@{}}
    #1
  \end{tabularx}
}
\makeatother

\makeatletter
\def\BState{\State\hskip-\ALG@thistlm}
\makeatother

\newcommand{\ceil}[1]{\lceil #1 \rceil}

\usepackage[compact]{titlesec}
\titlespacing{\section}{0pt}{3ex}{2ex}
\titlespacing{\subsection}{0pt}{2ex}{1ex}
\titlespacing{\subsubsection}{0pt}{0.5ex}{0ex}

\newtheorem{theorem}{Theorem}[section]

\newtheorem{corollary}{Corollary}[section]

\newtheorem{lemma}{Lemma}[section]
\newtheorem{claim}{Claim}

\theoremstyle{remark}

\makeatletter
\let\c@fconjecture\c@conjecture
\makeatother

\makeatletter
\let\c@fconj\c@conj
\makeatother

\def \eps {\varepsilon}

\newcommand{\ignore}[1]{}

\def \poly { \text{\rm poly~} }

\def \eps {\epsilon}

\title{Algorithms and Hardness for Diameter in Dynamic Graphs}

\author{Bertie Ancona \footnote{\texttt{bancona@mit.edu}} \\ MIT
\and Monika Henzinger \footnote{\texttt{monika.henzinger@univie.ac.at}, the research leading to these results has received funding
from the European Research Council under the European Community's Seventh Framework Programme (FP7/2007-2013) / ERC grant agreement No.~340506.} \\ University of Vienna
\and Liam Roditty \footnote{\texttt{liam.roditty@biu.ac.il}}\\ Bar Ilan University
\and Virginia Vassilevska Williams \footnote{\texttt{virgi@mit.edu}, supported by an NSF CAREER Award, NSF
Grants CCF-1417238, CCF-1528078 and CCF-1514339, a BSF Grant
BSF:2012338 and a Sloan Research Fellowship.}\\MIT
\and Nicole Wein \footnote{\texttt{nwein@mit.edu}, supported by an NSF Graduate Fellowship and NSF Grant CCF-1514339}\\ MIT}
\date{}

\begin{document}
\maketitle
\begin{abstract}
The diameter, radius and eccentricities are natural graph parameters. While these problems have been studied extensively, there are no known {\em dynamic} algorithms for them beyond the ones that follow from trivial recomputation after each update or from solving dynamic All-Pairs Shortest Paths (APSP), which is very computationally intensive. This is the situation for dynamic {\em approximation} algorithms as well, and even if only edge insertions or edge deletions need to be supported.

This paper provides a comprehensive study of the dynamic approximation of Diameter, Radius and Eccentricities, providing both conditional lower bounds, and new algorithms whose bounds are optimal under popular hypotheses in fine-grained complexity. Some of the highlights include: 
\begin{itemize}
\item Under popular hardness hypotheses, there can be no significantly better {\em fully dynamic approximation} algorithms than recomputing the answer after each update, or maintaining full APSP.

\item Nearly optimal {\em partially dynamic} (incremental/decremental) algorithms can be achieved via efficient reductions to (incremental/decremental) maintenance of Single-Source Shortest Paths. For instance, a nearly $(3/2+\eps)$-approximation to Diameter in directed or undirected graphs can be maintained decrementally in total time $m^{1+o(1)}\sqrt{n}/\eps^2$. This nearly matches the static $3/2$-approximation algorithm for the problem that is known to be conditionally optimal.  
\end{itemize}
\end{abstract}
\section{Introduction}
\input{intro-v.tex}

\input{preliminaries}
\input{2-ediameterradius.tex}
\input{incremental32.tex}

\input{incremental1eps.tex}

\paragraph{Acknowledgements} The authors would like to thank Roei Tov for discussions.

\bibliographystyle{plain}
\bibliography{referencesMH}
\end{document}

%% file: intro-v.tex
Computing the shortest paths distances between all pairs of vertices in a graph, the All-Pairs Shortest Paths (APSP) problem, has been studied since the beginning of computer science as a field. Nevertheless, the fastest known algorithms \cite{ryanapsp,PettieR05,Pettie04} for APSP in $n$-vertex $m$-edge graphs are only slightly faster (by $n^{o(1)}$ factors) than the simple $\tilde{O}(mn)$ time\footnote{$\tilde{O}$ notation supresses polylogarithmic factors.} algorithm running Dijkstra's algorithm from every vertex\footnote{after Johnson's trick to make the weights nonnegative}. For dense graphs with small integer weights there are improved algorithms \cite{seidel,Zwick02} using fast matrix multiplication \cite{vstoc12,legallmult}, but these algorithms are not faster than $mn$ for sparser graphs or when the weights can be large.

There are many important graph parameters that can be easily computed when all the distances are known. These include the \emph{eccentricity} of each vertex (the maximum length of a shortest path from the vertex to another vertex),  the graph \emph{diameter} (the maximum over all eccentricities), the \emph{radius} (the minimum over all eccentricities) and many more. These parameters are of particular importance in the analysis of social networks (e.g. \cite{albertdiam}), but also in graphs generated for entities such as images and search queries (and web pages).

Unfortunately, there are no significantly faster algorithms to compute these parameters than just solving APSP, and this is far from practical.
In many cases the analyzed networks are so large that even enumerating all pairs of vertices is prohibitively expensive. Thus, obtaining all pairwise distances is essentially impossible. For graph parameters, on the other hand, the output is a single number; in principle looking at all vertex pairs might not be necessary, and subquadratic time algorithms (in the number of vertices) might exist for sparse graphs (whereas quadratic time is necessary for APSP as this is the size of the output). The existence of such fast algorithms is an important, practically motivated question.

In recent years, much progress has been made in understanding the complexity of graph parameter computation. Results from fine-grained complexity give that even obtaining a $(3/2-\eps)$-approximation for graph diameter \cite{RV13} or radius \cite{AbboudWW16}, or a $(5/3-\eps)$-approximation of all eccentricities \cite{chechikdiam} (for $\eps>0)$ requires $n^{2-o(1)}$ time even in very sparse graphs, assuming the Strong Exponential Time Hypothesis (SETH) and related conjectures. Even stronger hardness results were obtained by Backurs et al.~\cite{diamstoc18}, altogether showing that most of the known algorithms for diameter, radius and eccentricities are conditionally optimal.

In addition to computing graph parameters in a static graph, a very natural goal is to maintain estimates of these parameters in a {\em dynamic} graph, where edges are inserted and deleted. In this setting, we would like to have a fast algorithm which preprocesses the given graph, and builds a data structure which can support edge updates efficiently and can answer queries about the parameter of interest in the current state of the data structure. This dynamic version of the problems is even more practically motivated, as real networks are naturally dynamic.

Unfortunately, the state of the art of dynamic algorithms for graph parameters such as Diameter is somewhat disappointing. The best known dynamic algorithms either just use the best known dynamic algorithms for APSP, or recompute the parameter estimate from scratch after each update. This leads to the following bounds:

(1) Demetrescu and Italiano~\cite{DemetrescuI04} obtained a fully dynamic exact APSP algorithm with an amortized update time of $\tilde{O}(n^2)$ and $O(1)$ query time; this is the best exact dynamic algorithm for the graph parameters as well. Abboud and Vassilevska W. \cite{AbboudW14} showed that under SETH, any $(4/3-\eps)$-approximation fully dynamic algorithm for diameter (for $\eps>0$) requires $n^{2-o(1)}$ amortized update or query time even in sparse graphs. Thus the APSP approach is conditionally optimal for fully dynamic ($4/3-\eps$)-approximate diameter algorithms. It is unclear however whether a $4/3$-approximation with better update time is possible, and whether the APSP bounds are best for Radius.

(2) By recomputing the parameter estimates after each query, one can maintain a 2-approximation for Diameter in directed graphs and Radius in undirected graphs in worst case $O(m+n)$ time per update, and a $(2+\eps)$-approximation for all Eccentricities in directed graphs for all $\eps$ in $\tilde{O}(m/\eps)$ time per update using the algorithm of Backurs et al. \cite{diamstoc18}. One can also maintain a $3/2$-approximation for Diameter and Radius, and a $5/3$-approximation of all Eccentricities in worst case time $\tilde{O}(m^{3/2})$ per update using the algorithms of \cite{RV13,chechikdiam}\footnote{An $\tilde{O}(m\sqrt n)$ time running time also exists if a slight additive error is allowed.}. More algorithms follow from the static results of Cairo, Grossi and Rizzi \cite{cgr}. Can any of these algorithms be improved or are they conditionally optimal? The only related lower bounds here are (a) by Henzinger et al.~\cite{HenzingerKNS15} which showed that under the Online Matrix Vector hypothesis (OMv), any fully dynamic Diameter algorithm that achieves a $(2-\eps)$-approximation for undirected {\em weighted} graphs, or any finite approximation in directed graphs needs $n^{0.5-o(1)}$ amortized update time and (b) by Henzinger et al.~\cite{HenzingerLNV17} which proved under the combinatorial Boolean Matrix Multiplication conjecture any fully dynamic Diameter or Eccentricity algorithm that achieves a $(4/3 - \eps)$-approximation in undirected {\em unweighted} graphs with $n^{3-o(1)}$ preprocessing time requires $n^{2 - o(1)}$ update or query time (and the same result for undirected {\em weighted} graph using the APSP conjecture). While these results give some limitation, they are far from tight.

The first contributions of our paper are strong conditional lower bounds for fully dynamic graph parameter estimation.
Our first result is a strengthening of the conditional lower bound for Diameter of \cite{AbboudW14}: we increase the approximation ratio from $(4/3-\eps)$ to $(3/2-\eps)$.

\begin{theorem}
Under SETH, every fully dynamic $(3/2-\eps)$-approximation algorithm for Diameter with polynomial preprocessing time requires $n^{2-o(1)}$ amortized update or query time in the word-RAM model of computation with $O(\log n)$ bit words, even for dynamic undirected unweighted graphs that are always sparse.
\end{theorem}
The same limitation applies for fully dynamic $(5/3-\eps)$-approximation algorithms for Eccentricities with polynomial preprocessing time, and for fully dynamic $(3/2-\eps)$-approximation algorithms for Radius with polynomial preprocessing time, under the related Hitting Set hypothesis.

These conditional lower bounds imply that the $\tilde{O}(m^{3/2})$ time estimation algorithms that recompute the answer from scratch are optimal in the sense that any improvement of the approximation factor causes the update time to grow to $n^2$, and Demetrescu and Italiano's algorithm achieves $\tilde{O}(n^2)$ update time even for the exact maintenance of APSP. 

We also show that recomputing a $2$-approximation from scratch in linear time is close to optimal under SETH. 
\begin{theorem}
Under SETH, any fully dynamic algorithm with polynomial preprocessing time that can maintain for $\eps>0$ any of the following in an $n$ node, $m$-edge undirected unweighted graph requires either $m^{1-o(1)}$ amortized update or query time, even when $m=\tilde{O}(n)$:
\begin{itemize}
\item a $(2-\eps)$-approximation of the eccentricity of a fixed vertex, or
\item a $(2-\eps)$-approximation of the Radius, or
\item a $(2-\eps)$-approximation of the Diameter.
\end{itemize}
\end{theorem}

This result significantly strengthens the OMv based lower bound of \cite{HenzingerKNS15}: the update time lower bound is now linear as opposed to $\sqrt n$, the result holds for unweighted graphs as well, it also holds for Radius and single-node eccentricity, and it has implications for partially dynamic algorithms with worst case time bounds, unlike the one of \cite{HenzingerKNS15}. Furthermore, the lower bound for the eccentricity of a fixed vertex is tight in the sense that  one can simply recompute the answer exactly from scratch in time $O(m)$. 

Much stronger lower bounds are possible for {\em directed} graphs. Our first hardness result for directed graphs is that under SETH and the Hitting Set hypothesis, respectively, nearly quadratic time is needed for Eccentricities and Radius, even for $(2-\eps)$-approximation. (Recall that for undirected graphs we could only show this for $(3/2-\eps)$-approximate Radius and $(5/3-\eps)$-approximate Eccentricities.) This means that for directed graphs, recomputing a $2$-approximation from scratch (in linear time) after each update is very much optimal -- for any better approximation one might as well use the exact dynamic APSP algorithms.

\begin{theorem}
Every fully dynamic algorithm with polynomial preprocessing time for $(2-\eps)$-approximate (for $\eps>0$) Eccentricities (under SETH) or Radius (under a version of the Hitting Set Hypothesis) in directed, unweighted graphs with $n$ vertices and $m=\tilde{O}(n)$ edges requires amortized update or query time $m^{2-o(1)}$.
\end{theorem}

Surprisingly, we also show conditionally that no {\em finite}
\footnote{A finite approximation algorithm is defined as an algorithm that can distinguish whether the value is finite or infinite.}
approximation can be maintained in {\em sublinear} time. Henzinger et al.~\cite{HenzingerKNS15}, building upon \cite{AbboudW14}, showed that any finite approximation for Diameter in directed graphs requires $m^{0.5-o(1)}$ time under the OMv Hypothesis. We strengthen the lower bound to linear, using a very natural hypothesis on the complexity of $k$-Cycle.

All known algorithms for detecting $k$-cycles in sparse directed graphs with $m$ edges run at best in time $m^{2-c/k}$ for various small constants $c$ \cite{YuZw04,AlYuZw97,patternscycles19}, even if you use powerful tools such as fast matrix multiplication. A natural hypothesis completely consistent with the state of the art of cycle detection is that one needs $m^{2-f(k)-o(1)}$ time to find a $k$-cycle, for some continuous (over the reals) $f(k)$ that goes to $0$ as $k$ goes to infinity. Let us call this the $k$-Cycle Hypothesis. We obtain:

\begin{theorem}
Under the $k$-Cycle Hypothesis, any fully dynamic algorithm with polynomial preprocessing time that can maintain a {\em finite} approximation for any of the following in an $n$ node, $m=\tilde{O}(n)$-edge directed unweighted graph requires either $m^{1-o(1)}$ amortized update or query time:
\begin{itemize}
\item the eccentricity of a fixed vertex, or
\item the Radius, or
\item the Diameter.
\end{itemize}
\end{theorem}

All known approaches to estimating graph proximity parameters such as the Diameter, at the very least require maintaining approximate distances from a single node, up to some distance. The conditional lower bounds above say that even if we only want to maintain an estimate of the largest distance from a fixed node, and even if that distance is never more than a constant, we still need linear update time. Thus, to have better than $2$ approximations of our undirected distance parameters or any finite approximation in the directed case that can be maintained in sublinear time we probably need to abandon our need for fully dynamic algorithms. We thus turn to {\em partially} dynamic algorithms that handle either only edge insertions (incremental) or only edge deletions (decremental).

Our conditional lower bounds for the fully dynamic setting also apply to incremental and decremental algorithms that have {\em worst case} update and query time guarantees. This is due to the nature of our reductions: they all produce an initial graph on which we perform update stages that only insert or only delete (we can choose which) a small batch of edges, ask a query and undo the changes just made, returning to the initial graph. An incremental/decremental algorithm can be used to implement such reductions by performing the deletions/insertions by rolling back the data structure. Because of this, we have very strong worst case lower bounds, and it makes sense to focus on amortized guarantees instead.

The strong conditional lower bounds in the static case, imply strong limitations for partially dynamic algorithms as well.
For {\em undirected} graphs, these limitations are as follows: Due to \cite{RV13,diamstoc18,chechikdiam}, under SETH, every incremental and decremental algorithm for Diameter in $n$ node undirected unweighted graphs requires total time at least $m^{3/2-o(1)}$ to maintain a $(8/5-\eps)$-approximation, and at least $m^{2-o(1)}$ total time to maintain a $(3/2-\eps)$-approximation for $\eps>0$ under $m=\tilde{O}(n)$ insertions or deletions. For Eccentricities, the static conditional lower bounds are slightly stronger. For partially dynamic algorithms they imply that under SETH, for every $k\geq 1$, maintaining a $((3k+2)/(k+2)-\eps)$-approximation for $\eps>0$ requires total time $m^{1+1/k}$ for $m=\tilde{O}(n)$ insertions or deletions.
For Radius, they just imply that under the Hitting Set hypothesis \cite{survey,AbboudWW16} maintaining a $(3/2-\eps)$-approximation requires total time $m^{2-o(1)}$ even in a sparse graph. 

For {\em directed} graphs, there are stronger lower bounds: maintaining a $(2-\eps)$-approximation for Radius and Eccentricities requires total time $m^{2-o(1)}$ under Hitting Set and SETH, respectively~\cite{AbboudWW16, diamstoc18}. 

The incremental lower bounds directly follow from the static ones by starting from an empty graph and inserting edges until we reach the graph from the static construction. The decremental lower bounds hold since the static lower bound instances are all subgraphs of the same global graph, 
 independent of the SAT/Hitting Set instance that the reduction is trying to solve; thus we start with the global graph and delete edges until reaching the graph from the static construction.

A natural question is, are these conditional lower bounds tight? Can one create partially dynamic algorithms that can achieve the same total runtime as the known static approximation algorithms? We give positive answers to these questions by developing new partially dynamic algorithms that are essentially optimal. Our algorithms are actually very efficient reductions to incremental and decremental single source shortest paths (SSSP), so that any improvement over dynamic SSSP would improve our parameter estimation algorithms. 

Let $D_0$ and $D_f$ be the initial and final values of the diameter, respectively. Let $T_{inc}(n,m,k,\epsilon)$ (resp., $T_{dec}(n,m,k,\epsilon)$) be the total time of an incremental (decremental) approximate SSSP algorithm from source $u$ that maintains an estimate $d'(u,v)$ for all $v$ such that if $d(u,v)\leq k$ then $(1-\epsilon)d(u,v)\leq d'(u,v)\leq d(u,v)$.
For directed graphs we assume that the approximate SSSP algorithm works in directed graphs, and for undirected graphs, the SSSP algorithm only needs to work in undirected graphs. Our black-box reductions can be summarized in the theorem below.

\begin{theorem}\label{thm:full} 
There is a Las Vegas randomized algorithm for incremental (resp., decremental) diameter in unweighted, directed graphs against an oblivious (resp., adaptive) adversary that given $\eps>0$, runs in total time $\tilde{O}(\max_{D_f\leq D'\leq D_0}\{T_{inc}(n,m,D',\eps)\frac{\sqrt{n/ D'}}{\eps^2}\})$ (resp., $\tilde{O}(\max_{D_0\leq D'\leq D_f}\{T_{dec}(n,m,D',\eps)\frac{\sqrt{n/ D'}}{\eps^2}\})$) with high probability, and maintains an estimate $\hat{D}$ such that $\frac{2(1-\eps)}{3}D-\frac{2}{3}\leq \hat{D}\leq D$ where $D$ is the diameter of the current graph. 
\end{theorem}

We obtain similar black-box reductions for nearly $(5/3+\eps)$-approximate Eccentricities and $(3/2+\eps)$-approximate Radius in undirected graphs.

Henzinger et al.~\cite{HenzingerKNFOCS14} obtained a ramdomized $(1+\eps)$-approximate decremental algorithm for SSSP in undirected unweighted graphs against an oblivious adversary with total expected update time $m^{1+o(1/\eps)}$.\footnote{The exact expected update time of this algorithm is $m^{1+O(\log^{5/4} ((\log n)/\epsilon)/\log^{1/4} n) } \log W$.} As an immediate corollary we obtain:

\begin{corollary}
There is a Las Vegas randomized algorithm for decremental diameter in unweighted, undirected graphs against an oblivious adversary that given $\eps>0$, runs in total time $m^{1+o(1/\eps)}\sqrt{n}/\eps^2$ in expectation, and maintains an estimate $\frac{2(1-\eps)}{3}D-\frac{2}{3}\leq \hat{D}\leq D$, where $D$ is the diameter of the current graph.
\end{corollary}
Due to lower bounds in the static setting described above, this result is conditionally optimal in terms of both running time and approximation factor except for a small loss in the approximation factor. A similar result hold for decremental undirected Radius with a conditionally essentially optimal approximation factor. A similar result also holds for Eccentricities, which is conditionally essentially optimal in terms of both running time and approximation factor.

For decremental algorithms in directed graphs and for incremental algorithms in undirected or directed graphs, the best known algorithms for SSSP up to distance $k$ are achieved by the Even and Shiloach Trees data structure \cite{EvenS81}, giving amortized update time $O(k)$.
\footnote{Although it is not published, it is likely that one can do better for undirected incremental graphs. In particular, it is widely believed that the known decremental SSSP algorithms for undirected graphs can be modified to the incremental setting with the same running time.}
Henzinger and King recognized that this data structure can be extended to directed graphs~\cite{HenzingerK95}. As a corollary we obtain:
\begin{corollary}
There is a Las Vegas randomized algorithm for incremental/decremental diameter in unweighted, directed graphs that given $\eps>0$, runs in total time $\tilde{O}(m\sqrt{nD_{\max}}/\eps^2)$ with high probability where $D_{\max}$ is the maximum diameter throughout the algorithm, and maintains an estimate $\hat{D}$ such that $\frac{2(1-\eps)}{3}D-1\leq \hat{D}\leq D$, where $D$ is the diameter of the current graph. The incremental algorithm works against an oblivious adversary and the decremental algorithm works against an adaptive adversary.
\end{corollary}

Similar results hold for Radius and Eccentricities but only for undirected graphs. Recall that static conditional lower bounds rule out such algorithms in directed graphs. 

The algorithms so far are all randomized. We present some deterministic incremental algorithms as well, again via a reduction to incremental SSSP. Let $D_0,D_f$ and $T_{inc}(n,m,k,\epsilon)$ be as before.

\begin{theorem}  
There is a deterministic algorithm  for incremental diameter in unweighted, directed  graphs that, for any   $\eps$ with  $0<\eps<2$, runs in total time $\tilde{O}(\max_{D_f\leq D'\leq D_0} \{(T_{inc}(n,m,D',\eps) +
m) n/(\eps^2 D') \})$, and maintains an estimate $\hat{D}$ such that $(1-\eps)D\leq \hat{D}\leq D$, where $D$ is the diameter of the current graph.
\end{theorem}

Using Even and Shiloach trees we obtain as a corollary a deterministic incremental $(1+\eps)$-approximation  algorithm for diameter with total update time $\tilde{O}(mn/\eps^2)$. The running time is essentially tight for $\eps<1/2$ by the SETH based quadratic lower bound for $(3/2-\delta)$-approximate static diameter \cite{RV13}. Similar algorithms with essentially tight running times hold for radius in directed graphs and eccentricities in directed, strongly connected graphs.

\subsection{Our techniques for partially dynamic algorithms}
Our partially dynamic nearly $3/2$-approximation algorithms for diameter and radius and our nearly $5/3$-approximation algorithm for eccentricities are based on known algorithms in the static setting~\cite{RV13,chechikdiam}. These static algorithms work by carefully choosing a set $U$ of vertices, performing SSSP from every vertex in $U$, and showing that at least one of these SSSP instantiations yields a good estimate for the parameter of interest. The set $U$ is chosen as follows. We pick a random sample $S$ of $\tilde{\Theta}(\sqrt{n})$ vertices and let $w^*$ be the vertex that is farthest from $S$; that is, $w^*$ is the vertex that maximizes $\min_{s\in S}d(w^*,s)$. Then, letting $N(w,\sqrt{n})$ be the closest $\sqrt{n}$ vertices to $w$, we set $U=S\cup\{w^*\}\cup N(w^*,\sqrt{n})$.

Adapting these static algorithms to the dynamic setting presents two main challenges:

(1) Firstly, given a set $S$ of vertices, the farthest vertex $w^*$ from $S$ can change over time. We wish to minimize the total number of vertices that we ever run dynamic SSSP from, as reinitializing dynamic SSSP from a new vertex is expensive. Suppose we run dynamic SSSP from every vertex in $N(w^*,\sqrt{n})$ at all times. Then, every time $w^*$ changes, we must reinitialize the dynamic SSSP data structure from $\sqrt{n}$ new vertices. If $w^*$ changes frequently, this is prohibitively slow. To overcome this issue, we show that it suffices to choose a vertex $w$ that \emph{approximates} $w^*$ (for a careful notion of approximation); and furthermore, by doing so we can limit the number of times we choose a new $w$.

Due to inherent differences between the incremental and decremental settings, we choose $w$ in different ways in the different settings. In the decremental setting, distances can only increase, so our current choice of $w$ can only become a poor approximation for $w^*$ if $d(w^*,S)$ increases. Then, we use the fact that $d(w^*,S)$ is monotonically increasing to bound the number of times we need to choose a new $w$.

The incremental setting is more involved. Since distances can only decrease, our current choice of $w$ becomes a poor approximation of $w^*$ if $d(w,S)$ decreases. A challenge arises because unlike $d(w^*,S)$, the distance $d(w,S)$ does not change monotonically. One can imagine a scenario in which whenever we choose a new $w$, an edge is added causing $d(w,S)$ to immediately decrease to 1, which mandates that we choose a new $w$. We address this challenge by carefully employing randomness against an oblivious adversary. We argue that by randomly sampling $w$ from a specifically chosen set of vertices, in expectation it will take a long time for $w$ to become a poor approximation for $w^*$.

(2) Secondly, we wish to apply a partially dynamic SSSP algorithm as a subroutine, however the state of such algorithms is 
much better for undirected graphs than directed graphs. 
For instance, for \emph{undirected} decremental graphs, there is a randomized $(1+\epsilon)$-approximate SSSP algorithm that runs amortized $m^{o(1)}$ time~\cite{HenzingerKNFOCS14} (and it is believed, but not published, than a similar result is possible for incremental graphs), while for incremental/decremental \emph{directed} graphs the best known algorithms for SSSP up to distance $k$ run in amortized time $O(k)$ \cite{EvenS81}. 
To address this discrepancy, we carefully exploit the fact that longer paths are easier to hit by randomly sampling: we augment the algorithm with an additional subsampling routine that quadratically decreases the dependence of the running time on the diameter $D$. 



%% file: preliminaries.tex
\section{Preliminaries}\label{sec:prelim}

Let $G=(V,E)$ be a graph, where $|V|=n$ and $|E|=m$. For every $u,v\in V$ let $d_G(u,v)$ be the length of the shortest path from $u$ to $v$. We omit the subscript when $G$ is clear from context. Let $N_{out}(v, s)$ (resp., $N_{in}(v, s)$) be the set of the $s$ closest outgoing (incoming) vertices of $v$, where ties are broken by taking the vertex with the smaller ID. The eccentricity $\varepsilon(v)$ of a vertex $v$ is defined as $\max_{u\in V} d(v,u)$. The diameter $D$ of a graph is $\max_{v\in V}\varepsilon(v)$. The radius $R$ of a graph is $\min_{v\in V}\varepsilon(v)$. 

\subsection{Algorithms}
For all of our algorithms for diameter, radius, and eccentricities in undirected graphs as well as diameter in directed graphs, we assume that the diameter is finite. One can easily check if this is the case by running a dynamic 
reachability algorithm from a single vertex.

For our partially dynamic algorithms, we let $D_0$ and $D_f$ be the initial and final values of the diameter, respectively. Similarly, $R_0$ and $R_f$ are the initial and final values of the radius, respectively.

The running times of our randomized algorithms are with high probability, which we take to mean with probability at least $1-1/n^c$ for all constants $c$. The running times of our algorithms is written in terms of $n$ and $m$, which refer to an upper bound on the number of vertices and edges, respectively, over the entire sequence of updates. That is, for incremental algorithms, the running time is written in terms of the final values of $n$ and $m$ and for decremental algorithms the running time written is in terms of the initial values of $n$ and $m$.

Each of our algorithms is written as a reduction to a black-box incremental or decremental approximation algorithm for \emph{truncated SSSP}; that is, SSSP which provides a distance estimate for all nodes whose distance from the source is at most a given value $k$.  For generality, our algorithms are written for directed graphs and use directed SSSP algorithms, however if the graph is undirected one can simply run an undirected SSSP algorithm instead. Let \emph{out-$\mathcal{A}_{inc}(u,k,\delta)$} (resp., \emph{out-$\mathcal{A}_{dec}(u,k,\delta)$}) be an incremental (resp., decremental) algorithm that maintains for all $v$ an estimate $d'(u,v)$ such that if $d(u,v)\leq k$ then $(1-\delta)d(u,v)\leq d'(u,v)\leq d(u,v)$. Analogously, let \emph{in-$\mathcal{A}_{inc}(u,k,\delta)$} (resp., \emph{in-$\mathcal{A}_{dec}(u,k,\delta)$}) be an incremental (decremental) algorithm that maintains an estimate $d'(v,u)$ for all $v$ such that if $d(u,v)\leq k$ then $(1-\delta)d(v,u)\leq d'(v,u)\leq d(v,u)$. We assume that after every update, these algorithms output all nodes whose distance estimate has changed. Let $T_{inc}(n,m,k,\delta)$ (resp., $T_{dec}(n,m,k,\delta)$) be the total time of out-$\mathcal{A}_{inc}(u,k,\delta)$ and in-$\mathcal{A}_{inc}(u,k,\delta)$ (resp., out-$\mathcal{A}_{inc}(u,k,\delta)$ and in-$\mathcal{A}_{inc}(u,k,\delta)$) (or the corresponding undirected algorithms, depending on the setting).

The running times of our algorithms are written as the maximum of an expression over all values of the diameter $D$ (or radius $R$) throughout the entire sequence of updates. Although the maximum value of $D$ and $R$ in a partially dynamic graph either occurs at the beginning or end of the update sequence, the maximum value of the running time expression could occur for any value of $D$ or $R$.


Suppose we run in-$\mathcal{A}_{inc}$,  in-$\mathcal{A}_{dec}$, out-$\mathcal{A}_{inc}$, or out-$\mathcal{A}_{dec}$ from a vertex $v$. Then, let $B_{out}(v,r)$ be the set of vertices $u$ with $d'(v,u)\leq r$. 

For a subset $S\subseteq V$ of vertices and a vertex $v \in V$ we define $d(S,v):=\min_{s \in S}d(s,v)$. Similarly, $d(v,S):=\min_{s \in S}d(v,s)$. When the algorithms call for an approximation $d'(S,v)$ of $d(S,v)$, we add a dummy vertex $x$ with an edge to every vertex in $S$ and run out-$\mathcal{A}_{inc}$ (or out-$\mathcal{A}_{dec}$) from $x$; let $d'(S,v)=d'(x,v)-1$. We define and maintain $d'(v,S)$ analogously by adding a dummy vertex with an edge from every vertex in $S$. 
\begin{claim}
For all $u\notin S$, $(1-2\delta)d(u,S) \leq d'(u,S) \leq d(u,S)$.
\end{claim}\label{claim:pre}
\begin{proof}
$(1-2\delta)d(u,S) =
(1-2\delta)(d(u,x)-1) = d(u,x)-\delta d(u,x)-1+\delta(2-d(u,x))\leq d(u,x)-\delta d(u,x)-1=(1-\delta)d(u,x)-1 \leq d'(u,x)-1 =d'(u,S)$ and $d'(u,S)=d'(u,x)-1 \leq d(u,x)-1 = d(u,S)$. 
\end{proof}

For all of our algorithms for diameter and eccentricities, the bulk of the argument is to prove a lemma of the following form: if one is given values $P'$ and $\eps$ such that $P'$ is at most the true value $P$ of the parameter of interest, then there is an algorithm that outputs an estimate $\hat{P}$ such that $\alpha(1-\eps)P'-\beta\ \leq \hat{P}\leq P$ for appropriate $\alpha$ and $\beta$. In Lemma~\ref{lem:max}, we show that a lemma of the above form suffices to prove our theorems. (In Lemma~\ref{lem:max}, the number $k$ of parameters is 1 for the case of diameter and $n$ for the case of eccentricities.)
Lemma~\ref{lem:max} also requires a fast constant-factor approximation for the parameter of interest in the static setting. Such algorithms exist for directed diameter and eccentricities in near-linear time.


\begin{lemma}\label{lem:max}
Let $\pi_1,\pi_2,\dots,\pi_k$ be a set of graph parameters (e.g. eccentricities). Suppose there is a static $\tilde{O}(T'(n,m))$ time algorithm that gives a constant-factor approximation for all $\pi_i$. 

Let $P_1,P_2,\dots,P_k$ be the dynamically changing values of  $\pi_1,\pi_2,\dots,\pi_k$, respectively. 
Suppose there is an algorithm $\mathcal{P}$ that given a partially dynamic graph and values $P'$ and $\eps>0$, runs in total time $T(n,m,P',\eps)$ where $T$ is a polynomial, and maintains a set of estimates $\hat{P}_1\leq P_1,\dots,\hat{P}_k\leq P_k$ such that for all $i$, if $P'\leq P_i$, then $\hat{P}_i\geq \alpha (1-\eps)P'-\beta$ for constants $\alpha$ and $\beta$. 

Let $P_{min}$ and $P_{max}$ be the minimum and maximum respectively over all $P_i$ over the entire sequence of updates. Then there is an algorithm $\mathcal{P'}$ that given a partially dynamic graph and $\eps>0$, runs in total time $\tilde{O}(T'(n,m)+\max_{P_{min}\leq P'\leq P_{max}}T(n,m,P',\eps)/\eps)$ and maintains estimates $\hat{P}_1,\dots,\hat{P}_k$ such that for all $i$, $\alpha(1-\eps)P_i-\beta\leq \hat{P}_i\leq P_i$.
\end{lemma}
\begin{proof}
Let $\eps'=\eps/2$. We run $\ell=\tilde{O}(1/\eps)$ instantiations of algorithm $\mathcal{P}$ with input values $P'=P'_1,P'_2,\dots,P'_\ell$ respectively and $\eps=\eps'$. In the next paragraph we will define the $P'_i$ so that $P'_0\leq P_{min}$ and $P'_0=\Omega(P_{min})$,  $P'_{\ell}\geq P_{max}$ and $P'_{\ell}=O(P_{max})$, and for all $i<\ell$, $P'_i=(1-\eps')P'_{i+1}$. This implies that $\ell\leq \log_{\frac{1}{1-\eps'}} n=\tilde{O}(1/\eps)$, so the total running time is $\tilde{O}(\max_{P'_0\leq P'\leq P'_{\ell}}T(n,m,P',\eps)/\eps)=\tilde{O}(\max_{P_{min}\leq P'\leq P_{max}}T(n,m,P',\eps)/\eps)$. 

Let $Q_{min}$ and $Q_{max}$ be the dynamically changing values of the minimum and maximum (respectively) of all $P_i$ seen so far. That is, at the end of the update sequence $Q_{min}=P_{min}$ and $Q_{max}=P_{max}$. To set the $P'_i$ initially, we run the static constant-factor approximation for all $\pi_i$ from the theorem statement. This allows us to set the $P'_i$'s to satisfy the constraints from the previous paragraph, but with $P_{min}$ and $P_{max}$ replaced by the initial values of $Q_{min}$ and $Q_{max}$ respectively. Throughout the execution, we maintain the constraints from the previous paragraph, but with $P_{min}$ and $P_{max}$ replaced by the current values of $Q_{min}$ and $Q_{max}$ respectively. That is, as $Q_{min}$ and $Q_{max}$ change, we may need to add new values $P'_i$ either above or below the existing values. To do this, we use the algorithm $\mathcal{P}$ to dynamically estimate $Q_{min}$ and $Q_{max}$ as follows. We maintain that for all $P_i$, there exist $P'_j$, $P'_{j+1}$ such that $\alpha (1-\eps)P'_j-\beta\leq \hat{P}_i\leq \alpha (1-\eps)P'_{j+1}-\beta$. When this inequality is not true for some $P_i$, we add new values of $P'_j$ either above or below the existing values as appropriate to make the inequality true (renumbering the $P'_i$ as appropriate). Given this inequality, the lemma statement implies that $\alpha (1-\eps)P'_j-\beta\leq P_i\leq P'_{j+1}$. Then since $P'_j$ and $P'_{j+1}$ differ by only a factor of $(1-\eps')$, this gives a constant-factor approximation for $Q_{min}$ and $Q_{max}$, which enables us to set the $P'_i$ to satisfy the required constraints.

Let $\hat{P}_i^j$ be the value $\hat{P}_i$ returned by $\mathcal{P}$ given the input parameter $P'_j$. Following every update, for all $i$, the return value $\hat{P}_i$ for $\mathcal{P'}$ is $\max_j \hat{P}_i^j$. Fix a point in the sequence of updates and fix an $i$. We will show that there exists $j$ such that $\alpha(1-\eps)P_i-\beta\leq \hat{P}_i^j\leq P_i$. The second inequality is directly implied by the lemma statement, so we prove the first. Let $j$ be the largest value such that $P'_j\leq P_i$, so $(1-\eps')P_i\leq P'_j$. Then the lemma statement implies that $\hat{P}_i^j\geq \alpha (1-\eps')P'_j-\beta\geq \alpha (1-\eps')^2P_i-\beta\geq \alpha(1-\eps)P_i-\beta$.  \end{proof}

Lemma~\ref{lem:min} is analogous to Lemma~\ref{lem:max}, but applies to minimization problems (such as radius) instead of maximization problems (such as diameter and eccentricities). 

\begin{lemma}\label{lem:min}


Let $\pi_1,\pi_2,\dots,\pi_k$ be a set of graph parameters. Suppose there is a static $\tilde{O}(T'(n,m))$ time algorithm that gives a constant-factor approximation for all $\pi_i$. 

Let $P_1,P_2,\dots,P_k$ be the dynamically changing values of  $\pi_1,\pi_2,\dots,\pi_k$, respectively. 
Suppose there is an algorithm $\mathcal{P}$ that given a partially dynamic graph and values $P'$ and $\eps>0$, runs in total time $T(n,m,P',\eps)$ where $T$ is a polynomial, and maintains a set of estimates $\hat{P}_1\geq P_1,\dots,\hat{P}_k\geq P_k$ such that for all $i$, if $P'\geq P_i$, then $\hat{P}_i\leq \alpha (1+\eps)P'+\beta$ for constants $\alpha$ and $\beta$. 

Let $P_{min}$ and $P_{max}$ be the minimum and maximum respectively over all $P_i$ over the entire sequence of updates. Then there is an algorithm $\mathcal{P'}$ that given a partially dynamic graph and $\eps>0$, runs in total time $\tilde{O}(T'(n,m)+\max_{P_{min}\leq P'\leq P_{max}}T(n,m,P',\eps)/\eps)$ and maintains estimates $\hat{P}_1,\dots,\hat{P}_k$ such that for all $i$, $P_i\leq \hat{P}_i\leq \alpha(1+\eps)P_i+\beta$.

\end{lemma}
\begin{proof}
The proof is identical to the proof of Lemma~\ref{lem:max} except we set $\eps'=\eps/3$ and the last paragraph is replaced with the following:

Let $\hat{P}_i^j$ be the value $\hat{P}_i$ returned by $\mathcal{P}$ given the input parameter $P'_j$. Following every update, for all $i$, the return value is $\min_j \hat{P}_i^j$. Fix a point in the sequence of updates and fix an $i$. We will show that there exists $j$ such that $P_i\leq \hat{P}_i^j\leq \alpha(1+\eps)P_i+\beta$. The first inequality is directly implied by the lemma statement, so we prove the second. Let $j$ be the smallest value such that $P'_j\geq P_i$, so $P'_j\leq \frac{P_i}{1-\eps'}$. Then the lemma statement implies $\hat{P}_i^j\leq \alpha (1+\eps')P'_j+\beta\leq \alpha (\frac{1+\eps'}{1-\eps'})P_i+\beta\leq \alpha(1+\eps)P_i+\beta$.  \end{proof}

\subsection{Lower bounds}
Let $k\geq 2$. The $k$-Orthogonal Vectors Problem ($k$-OV) is as follows: given $k$  sets $S_1,\ldots,S_k$, where each $S_i$ contains $n$ vectors in $\{0,1\}^d$, determine whether there exist $v_1\in S_1,\ldots,v_k\in S_k$ so that their generalized inner product is $0$, i.e. $\sum_{i=1}^d \prod_{j=1}^k v_j[i]=0$. The $k$-OV Hypothesis is that $k$-OV requires $n^{k-o(1)}$ time in the word-RAM model of computation with $O(\log n)$ bit words, even for randomized algorithms.

The {\em unbalanced} version of $k$-OV has the sets $S_i$ potentially have different sizes, $|S_i|=n_i$. The unbalanced $k$-OV Hypothesis is that unbalanced $k$-OV requires $(\prod_i n_i)^{1-o(1)}$ time. When each $n_i$ is polynomial in $n$, the unbalanced $k$-OV Hypothesis is known to be equivalent to the $k$-OV Hypothesis.

R. Williams~\cite{TCS05} (see also \cite{survey}) showed that if for some $\eps>0$ there is an $n^{k-\eps}\poly(d)$ time algorithm for $k$-OV, then CNF-SAT on formulas with $N$ variables and $m$ clauses can be solved in $2^{N(1-\eps/k)}\poly(m)$ time. In particular, 
such an algorithm would contradict
the Strong Exponential Time Hypothesis (SETH) of Impagliazzo, Paturi and Zane~\cite{ipz2} which states that for every $\eps>0$ there is a $K$ such that $K$-SAT on $N$ variables cannot be solved in $2^{(1-\eps)N}\poly N$
time (say, on a word-RAM with $O(\log N)$ bit words) even by randomized algorithms. Thus SETH implies the $k$-OV Hypothesis for all constants $k$. Each of our lower bounds conditional upon SETH is a reduction from either unbalanced 2-OV or unbalanced 3-OV.

The Hitting Set (HS) problem \cite{AbboudWW16} is as follows: given two sets $U$ and $V$ of $n$ vectors each in $\{0,1\}^d$, is there some $u\in U$ so that for all $v\in V$, $u\cdot v\neq 0$? The HS Hypothesis states that HS requires $n^{2-o(1)}$ time in the word-RAM model with $O(\log n)$ bit words, even for randomized algorithms.

We introduce the unbalanced version of HS, for three unbalanced sets. Unbalanced $3$-HS is the problem, given $U,V,W\subseteq \{0,1\}^d$ with $|U|=n, |V|=n^a, |W|=n^b$ for constants $a,b>0$, are there $u\in U,w\in W$ so that for all $v\in V$, $u\cdot v\cdot w \neq 0$? This is in similar spirit to unbalanced $3$-OV.

The unbalanced 3-HS Hypothesis is that unbalanced 3-HS requires $(|U|\cdot |V|\cdot |W|)^{1-o(1)}= n^{1+a+b-o(1)}$ time in the word-RAM model with $O(\log n)$ bit words, even for randomized algorithms. Due to its similarity to $3$-OV and the lack of good algorithms, the 3-HS Hypothesis is believable. Refuting it would imply some very interesting improved algorithms for a balanced variant of Quantified Boolean Formulas with 2 quantifiers \cite{openproblems}.

As mentioned in the introduction, the $k$-Cycle Hypothesis is that in the word-RAM model with $O(\log n)$ bit words, any possibly randomized algorithm needs $m^{2-f(k)-o(1)}$ time to find a $k$-cycle in an $m$-edge graph, for some continuous (over the reals) $f(k)$ that goes to $0$ as $k$ goes to infinity. The Hypothesis is completely consistent with the state of the art $k$-Cycle algorithms (e.g. \cite{YuZw04,AlYuZw97,patternscycles19}).




%% file: 2-ediameterradius.tex
\section{Lower Bounds}

\subsection{Base graph $G_{\delta}$ and reduction outline}
\label{basegraph}

Most of our reductions based on SETH and the 3-HS hypothesis begin with a variation of the same undirected, unweighted graph $G_{\delta}$ and proceed in the same way. This construction is inspired by that of \cite{AbboudW14}, which gives a quadratic lower bound for a $(4/3-\eps)$-approximate diameter. A similar construction is also used in~\cite{abboud2016near} and~\cite{bringmann2018note}. We will describe $G_{\delta}$ and then how we can use $G_{\delta}$ to show lower bounds. In most of our SETH or 3-HS hypothesis reductions below, we will only describe how the reduction differs.

\subsubsection*{$G_\delta$ Initialization}
We will generally assume for the sake of contradiction that a dynamic $(\alpha-\eps)$-approximation algorithm exists (for the appropriate value of $\alpha$) for Diameter, Radius, or Eccentricities with preprocessing time $n^t$, amortized update time $n^{2-\eps'}$ and query time $n^{2-\eps'}$, for positive numbers $\eps,\eps',$ and $t$. We define $\delta = \frac{1-\eps'}{t}$. Our construction also includes a parameter $a$ that is defined separately for each individual construction such that the supposed $(\alpha-\eps)$-approximation algorithm may distinguish between the existence or non-existence of an orthogonal triple or hitting set in each stage.

We begin with an instance of 3-OV or 3-HS with vector sets $U$, $V$, and $W$ such that $|U| = |V| = N^\delta$ and $|W| = N^{1-2\delta}$. We first discard degenerate vectors and coordinates: any coordinates that are 0 for every vector in $U$, every vector in $V$, or every vector in $W$, and any vectors that are all zeroes. Note that removing degenerate coordinates does not change the correct output value. If there is a degenerate vector in a 3-OV instance or in $U$ or $W$ in a 3-HS instance, the correct output value is always ``yes", while if there is a degenerate vector in $V$ in a 3-HS instance, removing it does not change the correct output value. 

The construction of $G_\delta$ is shown in Figure~\ref{figGdelta} and described here. For each coordinate $c$, create two nodes, and denote one by $c_{U}$ and the other by $c_{V}$. We denote the two sets of coordinate nodes by $C_U$ and $C_V$ respectively. Next, create a path of length $a$ for each vector $u\in U$; denote the start by $u^0$ and the end by $u^a$, and each node at distance $i$ from $u^0$ by $u^i$. We make a similar path for each $v\in V$. Finally, we encode each vector $u\in U$ in the graph by adding a path of length $a$ between $u^a$ and $c_U$ if $u[c] = 1$, and encode each $v\in V$ by adding a path of length $a$ between $c_V$ and $v^0$ if $v[c] = 1$.

\subsubsection*{Stages}
We proceed in $N^{1-2\delta}$ stages, one for each element $w\in W$. Generally, we will add an edge to $G_\delta$ between $c_U\in C_U$ to $c_V\in C_V$ if $w[c] = 1$, for all coordinates $c$. Then, we will make a query to an algorithm we assume exists for the sake of contradiction. Depending on the result of the query, we will either detect an orthogonal vector triple or hitting set and halt, or we will undo our edge additions for this stage and continue to the next $w$. 

Each stage may be modified to be decremental by initializing $G_\delta$ with edges between each $c_U,c_V$ pair and removing the excess edges during each stage. 

\begin{figure}[h]
  \centering
  \includegraphics[width=.8\linewidth]{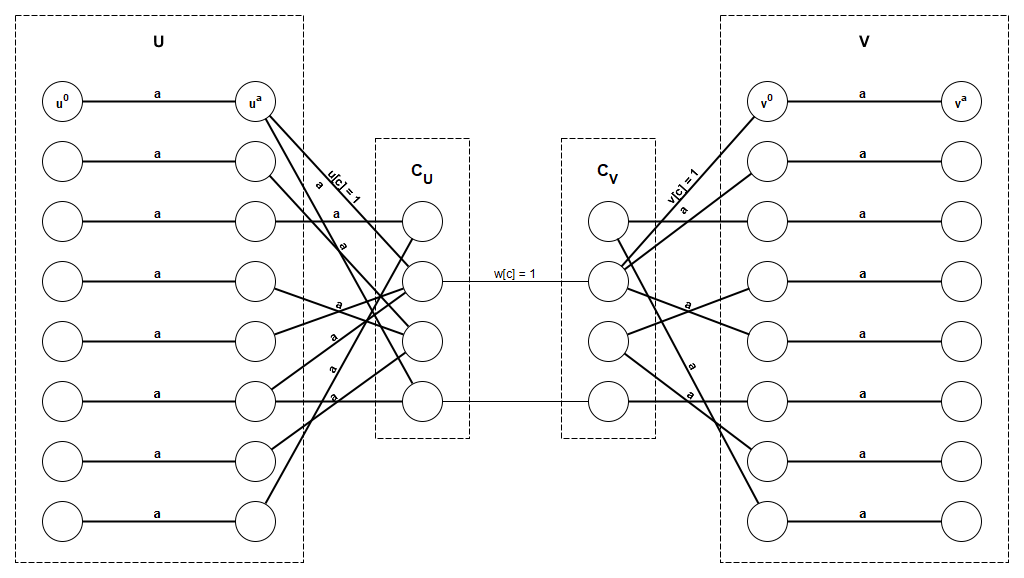}
  \caption{Sketch of $G_\delta$ construction. Bold edges represent paths, whose labels denote their length.}
  \label{figGdelta}
\end{figure}

\subsubsection*{Analysis}
We prove correctness separately for each individual reduction.

The running time analysis is the same for all of the reductions from 3-OV and 3-HS. The preprocessing time of the algorithm for a graph of size $N^\delta$ is $(N^\delta)^t = N^{1-\eps'}$. Each update or query takes amortized $O((N^\delta)^{2-\eps'})$ time, so after $N^{1-2\delta}$ stages wherein we make $\tilde{O}(1)$ updates and queries, the total amortized time is $\tilde{O}(N^{1-\delta\eps'})$. This gives an algorithm for 3-OV or 3-HS in $\tilde{O}(N^{1-\eps'}+N^{1-\delta\eps'}) = O((|U||V||W|)^{1-\eps''})$ for a positive constant $\eps''$, refuting SETH or the 3-HS hypothesis.

\subsection{$(3/2-\eps)$-approximation requires quadratic update time}

In this section we give a quadratic lower bound per update for $(3/2-\eps)$-approximation for diameter and radius, and $(5/3-\eps)$-approximation for eccentricities on undirected, unweighted graphs.

\subsubsection{Diameter}

\begin{theorem}
\label{thm:32lb}
Let $t$, $\eps$, and $\eps'$ be positive constants. SETH implies that there exists no fully dynamic algorithm for $(3/2-\eps)$-approximate Diameter on undirected, unweighted graphs with $n$ vertices and $\tilde{O}(n)$ edges, which has preprocessing time $p(n)=O(n^t)$, amortized update time $u(n)=O(n^{2-\eps'})$, and amortized query time $q(n)=O(n^{2-\eps'})$. 

The same result holds for the incremental and decremental settings but for worst-case update and query times.
\end{theorem}

\begin{proof}[Proof of Theorem~\ref{thm:32lb}]$ $

\paragraph{Construction}
\subparagraph{Initialization}

Let $a=\ceil{\frac{1-2\eps}{8\eps}}+1$.

We begin with an instance of 3-OV and construct a graph $G$ by first creating $G_\delta$ from section \ref{basegraph}. We add two nodes $x$ and $y$. For each node $u^a$, add a path of length $a$ between $x$ and $u^a$, and for each node $v^0$, add a path of length $a$ between $y$ and $v^0$.

\subparagraph{Stages} We proceed in $N^{2-\delta}$ stages, one for each element $w\in W$. For the current $w$, for each coordinate $c$ where $w[c]=1$, we add an edge between $c_U$ and $c_V$. We then query the diameter of $G$. We will show that if there is an orthogonal triple that includes $w$, then the diameter is least $6a+1$, and otherwise, the diameter is at most $4a+1$. A $(3/2-\eps)$-approximation algorithm for diameter distinguishes between these two cases and can thus detect an orthogonal triple that includes $w$ if one exists. If such an orthogonal triple is not detected, we undo the edge additions for the stage and continue to the next $w$.

\begin{figure}[h]
  \centering
  \includegraphics[width=.8\linewidth]{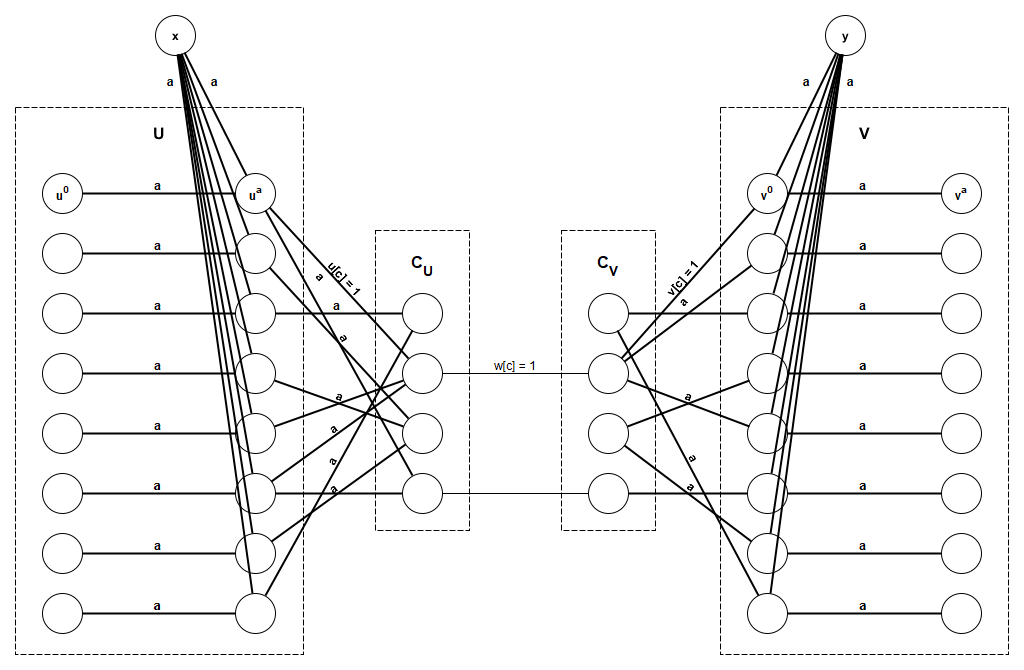}
  \caption{Sketch of Theorem \ref{thm:32lb} Construction. Bold edges represent paths, whose labels denote their length.}
\end{figure}


\paragraph{Correctness}
If for the current stage, for all $u\in U$, $v\in V$, $u\cdot v\cdot w\not= 0$, then for each $u,v$ there exists a coordinate $c$ such that $u[c]=v[c]=w[c]=1$. Thus, for all $u$ and $v$, $d(u^a,v^0)= 2a+1$. Also, for all $u,u'\in U$, $d(u^a,u'^a)\leq 2a$. We note that every vertex in the graph is of distance at most $a$ from \emph{some} vertex $u^a$ or $v^0$. Thus, the diameter is at most $4a+1$.

Suppose for the current stage there exist $u\in U$, $v\in V$ such that $u\cdot v\cdot w =0$. Fix $u$ and $v$. We claim that $d(u^0,v^a)\geq 6a+1$. The only paths between $u^0$ and $v^a$ go through $u^a$ and $v^0$. There does not exist a coordinate such that $u[c]=v[c]=w[c]=1$, so every path between $u^a$ and $v^0$ must visit a vertex $u'^a$ or $v'^0$ (for $u\in U$, $u'\neq u$, $v'\in V$, $v'\neq v$), which are of distance $2a$ from $u^a$ and $v^0$, respectively. Thus, $d(u^a,v^0)\geq 4a+1$ so $d(u^0,v^a)\geq 6a+1$.
\end{proof}
\subsubsection{Radius}
\begin{theorem}
\label{thm:32radlb}
Let $t$, $\eps$, and $\eps'$ be positive constants. The 3-HS hypothesis implies that there exists no fully dynamic algorithm for $(3/2-\eps)$-approximate Radius on undirected, unweighted graphs with $n$ vertices and $\tilde{O}(n)$ edges, which has preprocessing time $p(n)=O(n^t)$, amortized update time $u(n)=O(n^{2-\eps'})$, and amortized query time $q(n)=O(n^{2-\eps'})$. 

The same result holds for the incremental and decremental settings but for worst-case update and query times.
\end{theorem}

\begin{proof}[Proof of Theorem~\ref{thm:32radlb}]$ $

\paragraph{Construction}
\subparagraph{Initialization}

Let $a=\ceil{\frac{1-2\eps}{8\eps}}+1$.

We begin with an instance of 3-HS and construct a graph $G$ by first creating two copies of the graph from Theorem \ref{thm:32lb}; call these copies $G_{left}$ and $G_{right}$. $G$ will be the graph composed of $G_{left}$ and $G_{right}$, where for each $u\in U$, we merge the vertex $u_{left}^0$ in $G_{left}$ with the vertex $u_{right}^0$ in $G_{right}$.

\subparagraph{Stages} We proceed in $N^{2-\delta}$ stages, one for each element $w\in W$. For the current $w$, for each coordinate $c$ where $w[c]=1$, we add an edge between $c_{U,left}$ and $c_{V,left}$, and the same edge on the right side. We then query the radius of $G$. We will show that if there is a hitting set that includes $w$, then the radius is most $4a+1$, and otherwise, the radius is at least $6a+1$. A $(3/2-\eps)$-approximation algorithm for radius distinguishes between these two cases and can thus detect a hitting set that includes $w$ if one exists. If such a hitting set is not detected, we undo the edge additions for the stage and continue to the next $w$.

\begin{figure}[h]
  \centering
  \includegraphics[width=\linewidth]{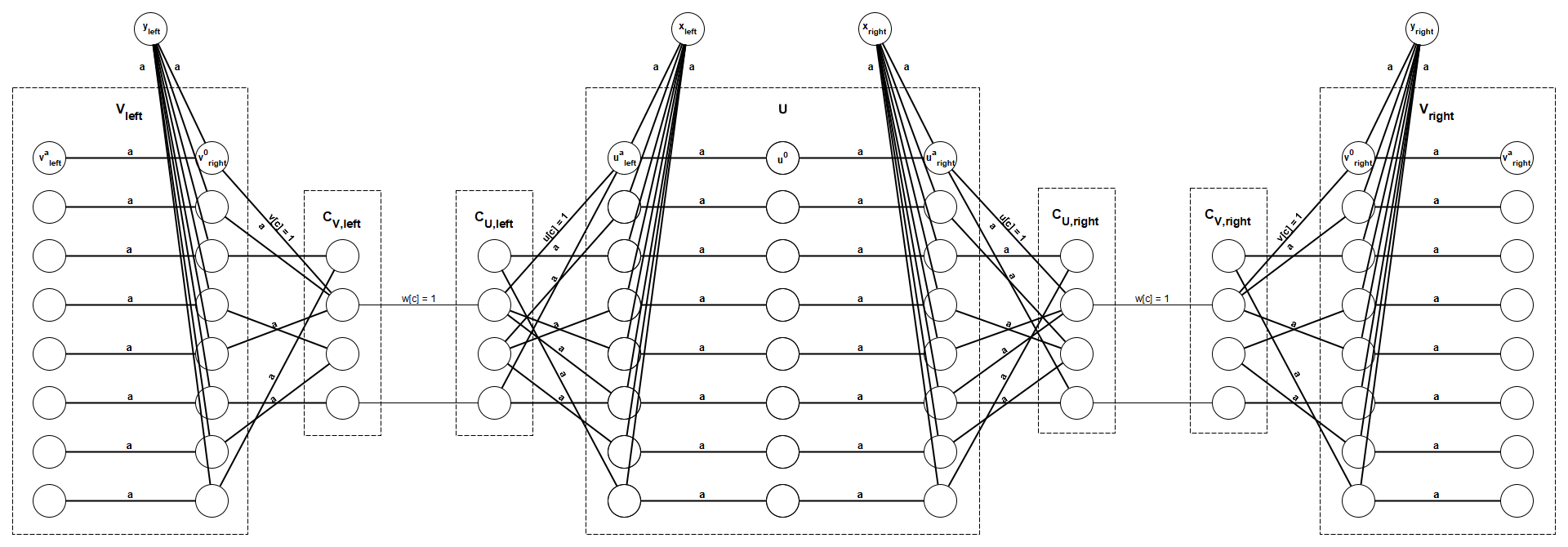}
  \caption{Sketch of Theorem \ref{thm:32radlb} Construction. Bold edges represent paths, whose labels denote their length.}
\end{figure}

\paragraph{Correctness}

Suppose there is a hitting set $\{u,w\}$ for the current stage. We will show that the node $u^0$ is of distance at most $4a+1$ from every other node. The node $u^0$ can reach all $u'^a$ nodes where $u'\in U$ on either side of the graph in $3a$ steps, via $x_{left}$ or $x_{right}$. 
Since $u$ and $w$ define the hitting set, $u^0$ has a path of length $3a+1$ to all $v^0$ on both sides of the graph, via the nodes for the coordinates equal to 1 in $u^0$, $w$, and each $v^0$. 

We note that all nodes are at distance $a$ from some node $u'^a$ where $u'\in U$, or some $v^0$ either on the left or right side of the graph. Thus, the previous paragraph implies that $u^0$ can reach every node in $4a+1$ steps.

Suppose there is no hitting set involving $w$ and any $u$. Fix a side of the graph (we will omit subscripts \emph{left} and \emph{right}). The only paths between $u^0$ and any $v^a$ go through $u^a$ and $v^0$ (on the appropriate side of the graph). For all $u$, there exists a $v$ such that there does not exist a coordinate $c$ with $u[c]=v[c]=w[c]=1$. Fix $u$ and $v$. Every path between $u^a$ and $v^0$ must visit some node $u'^a$ or some node $v'^0$ (for $u\in U$, $u'\neq u$, $v'\in V$, $v'\neq v$), which are of distance $2a$ from $u^a$ and $v^0$, respectively. Thus, $d(u^a,v^0)\geq 4a+1$ so $d(u^0,v^a)\geq 6a+1$. 
Thus, all $u^0$ have eccentricity at least $6a+1$. All other nodes have a higher eccentricity than some $u^0$, because they must travel via some $u^0$ to the other side of the graph and are thus farther from $u^0$'s farthest node (of which it has at least one on either side of the graph, by symmetry). Thus, the radius must be at least $6a+1$.
\end{proof}
\subsubsection{Eccentricities}
\begin{theorem}
\label{thm:53ecclb}
Let $t$, $\eps$, and $\eps'$ be positive constants. SETH implies that there exists no fully dynamic algorithm for $(5/3-\eps)$-approximate all Eccentricities on undirected, unweighted graphs with $n$ vertices and $\tilde{O}(n)$ edges, which has preprocessing time $p(n)=O(n^t)$, amortized update time $u(n)=O(n^{2-\eps'})$, and amortized query time $q(n)=O(n^{2-\eps'})$. 

The same result holds for the incremental and decremental settings but for worst-case update and query times.
\end{theorem}

\begin{proof}[Proof of Theorem~\ref{thm:53ecclb}]$ $

\paragraph{Construction}
\subparagraph{Initialization}

Let $a=\frac{7-6\eps}{9\eps}$.

We begin with an instance of 3-OV and construct a graph $G$ by first creating $G_\delta$ from section \ref{basegraph}. We then add a node $x$, with an edge between $x$ and each node $u^0$ for all $u\in U$.

\subparagraph{Stages} We proceed in $N^{2-\delta}$ stages, one for each element $w\in W$. For the current $w$, for each coordinate $c$ where $w[c]=1$, we add an edge between $c_U$ and $c_V$. We then query all eccentricities of $G$. We will show that if there is an orthogonal triple $u,v,w$, then $u^a$ will have eccentricity at least $5a+1$; otherwise, for all $u\in U$, $u^a$ will have eccentricity at most $3a+2$. A $(5/3-\eps)$-approximation algorithm for all eccentricities distinguishes between these two cases and can thus detect an orthogonal triple that includes $w$ if one exists.

If such an orthogonal triple is not detected, we undo the edge additions for the stage and continue to the next $w$.

\begin{figure}[h]
  \centering
  \includegraphics[width=.8\linewidth]{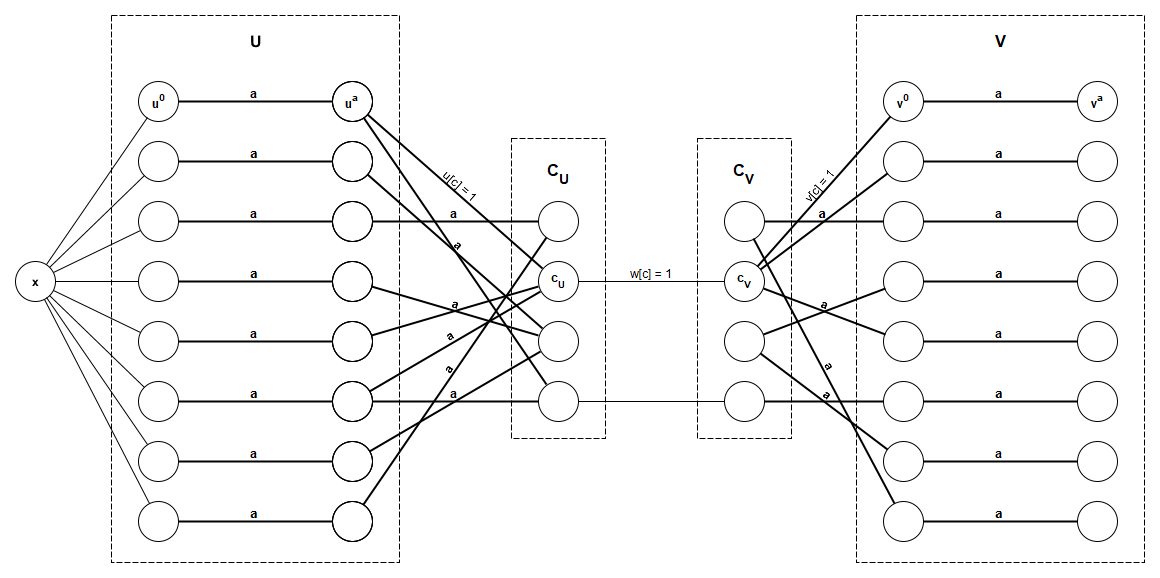}
  \caption{Sketch of Theorem \ref{thm:53ecclb} Construction. Bold edges represent paths, whose labels denote their length.}
\end{figure}


\paragraph{Correctness}
If for the current stage, for all $u\in U$, $v\in V$, $u\cdot v\cdot w\not= 0$, then for each $u,v$ there exists a coordinate $c$ such that $u[c]=v[c]=w[c]=1$. Thus, for all $u$ and $v$, $d(u^a,v^0) = 2a + 1$. Also, for each $u,u'\in U$, $d(u^a,u'^a)\leq 2a+2$ via a path through $x$. We note that all nodes are within distance $a$ of \emph{some} node $u^a$ or $v^0$, except $x$ which is at most distance $a+1$ from any $u^a$, so the eccentricity of all $u^a$ nodes is at most $3a+2$.

Suppose for the current stage there exist $u\in U$, $v\in V$ such that $u\cdot v\cdot w =0$. Fix $u$ and $v$. We claim that $d(u^a,v^a)\geq 5a+1$. Since there is no coordinate $c$ for which $u[c]=v[c]=w[c]=1$, the path from $u^a$ to $v^a$ must visit some $u'^a$ or $v'^a$ (for $u\in U$, $u'\neq u$, $v'\in V$, $v'\neq v$). This means adding a detour of at least $2a$ from a direct path of length $3a+1$, giving $d(u^a,v^a)\geq 5a+1$. Thus, the eccentricity of $u^a$ must be at least $5a+1$.
\end{proof}

\subsection{$(2-\eps)$-approximation requires linear update time}

In this section we give a linear lower bound per update for $(2-\eps)$-approximation for diameter, radius, and fixed-vertex eccentricity.

\begin{theorem}
\label{thm:2approxdiamrad}
Let $t$, $\eps$, and $\eps'$ be positive constants. SETH implies that there exists no fully dynamic algorithm for $(2-\eps)$-approximate Diameter, Radius, or fixed-vertex Eccentricity on undirected, unweighted graphs with $n$ vertices and $\tilde{O}(n)$ edges, which has preprocessing time $p(n)=O(n^t)$, amortized update time $u(n)=O(n^{1-\eps'})$, and amortized query time $q(n)=O(n^{1-\eps'})$.

The same result holds for the incremental and decremental settings but for worst-case update and query times.
\end{theorem}

\begin{proof}[Proof of Theorem~\ref{thm:2approxdiamrad}]$ $

\subsubsection{Construction}
\subsubsection*{Initialization}

Let $a = \ceil{\frac{2-\eps}{2\eps}}+1$.

We begin with an instance of 2-OV with vector sets $U$ and $V$ of vectors, with $|U| = N^\delta$ and $|V| = N^{1-\delta}$. We create a graph $G$ as follows. Add a node $s$ and for each coordinate $c$, create two paths of length $2a$ beginning at $s$, and denote the endpoints of the paths as $c_{left}$ and $c_{right}$.

Next, create two paths of length $a$ for each vector $u\in U$. Denote the endpoints of one path by $u_{left}^0$ and $u_{left}^a$, and the endpoints of the other by $u_{right}^0$ and $u_{right}^a$. Finally, we encode each vector $u\in U$ in the graph by connecting $u_{left}^a$ to $c_{left}$ with a path of length $a$ if $u[c] = 1$, and doing the same on the right side of $G$. If $u$ has no coordinates equal to 1, then we may report that there is an orthogonal pair and halt; thus there will be no disconnected nodes in $G$.

Essentially, we have just created a modified version of $G_\delta$, but we encoded $U$ twice instead of encoding $V$.

\subsubsection*{Stages}

We proceed in $N^{1-\delta}$ stages, one for each element $v\in V$. For the current $v$, for each coordinate $c$ where $v[c] = 1$, we add edges $(s,c_{left})$ and $(s,c_{right})$. 

We then query the diameter or the radius of $G$ or the eccentricity of $s$. We will show that the eccentricity of $s$ is always equal to the radius, and we will show that if the diameter is least $8a$ or the radius is at least $4a$, then there is an orthogonal pair $u,v$; otherwise, the diameter is at most $4a+2$ or the radius is at most $2a+1$. We have set $a$ such that a $(2-\eps)$-approximation algorithm for diameter or radius or eccentricity of $s$ distinguishes between these two cases and can thus detect an orthogonal pair $u, v$ if one exists. If the query does not detect such an orthogonal pair, we undo the edge additions for the stage and continue to the next $v$.

We can modify the stage to be decremental by beginning with edges from $s$ to all nodes $c_{left}$ and $c_{right}$, and removing the excess edges each stage.

\begin{figure}[h]
  \centering
  \includegraphics[width=.8\linewidth]{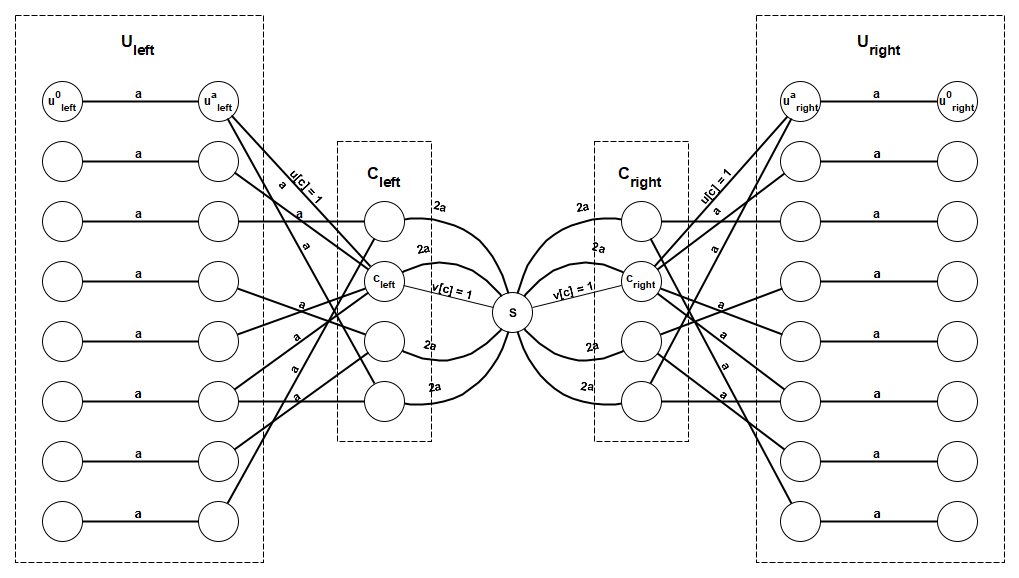}
  \caption{Sketch of Theorem \ref{thm:2approxdiamrad} Construction. Bold edges represent paths, whose labels denote their length.}
\end{figure}

\subsubsection{Analysis}

\subsubsection*{Correctness}

We first claim that the node $s$ is always the center of $G$, so the eccentricity of $s$ and the radius of $G$ are always equal. Let $x_{right}$ be the node farthest from $s$ on the right side of $G$. Since the graph is exactly symmetrical, the counterpart $x_{left}$ of $x_{right}$ is such that $d(s,x_{right})=d(s,x_{left})$. Any node $y$ to the left of $s$ must pass through $s$ to reach $x_{right}$, so $y$ must have a higher eccentricity than $s$ because it is farther from the node farthest from $s$. Symmetrically, any node $y$ to the right of $s$ must pass through $s$ to reach $x_{left}$, so $y$ must have a higher eccentricity than $s$ because it is farther from the node farthest from $s$. 

If for the current stage, for all $u$, $u\cdot v \neq 0$, then for each $u$ there must be some coordinate $c$ such that $u[c]=v[c]=1$. Then there is a path of length 1 from $s$ to $c_{left}$, and a path of length $a$ from $c_{left}$ to $u_{left}^a$, for all $u$. The same is true on the right side. Then since all nodes except $s$ are of distance at most $a$ from a node $u_{left}^{'a}$ or $u_{right}^{'a}$ for some $u'\in U$, all nodes are accessible in at most $2a+1$ steps from $s$. This means that the radius and eccentricity of $s$ is $2a+1$, and the diameter is at most $4a+2$.

If for the current stage there is some $u$ such that $u\cdot v = 0$, then there is no direct path from $s$ to $u^0$ on either side via a vector coordinate $c$ and $u^a$. A path via a different $u'_a$ would be of length at least $4a+1$, because returning to a $c'$ where $u[c']=1$ would cost an additional $2a$ from the direct path. The shortest path would thus be along the length-$2a$ path from $s$ to $c'$, giving $d(s,u^0) = 4a$. The radius and eccentricity of $s$ must be at least $4a$ and diameter must be at least $8a$, because $d(u_{left}^0,u_{right}^0) = d(s, u_{left}^0) + d(s,u_{right}^0) = 4a + 4a = 8a$.

\subsubsection*{Running time}

We assume for the sake of contradiction that the algorithm of Theorem \ref{thm:2approxdiamrad} exists. Let $n=N^\delta$ be the size of $G$. We have that $u(n) = q(n) = O((N^{\delta})^{1-\eps'})$. After initialization and $|V| = N^{1-\delta}$ stages, the total update and query time is then at most $\tilde{O}(N^{1-\delta\eps'})$. The preprocessing time $p(n)$ for the algorithm on $G$ is $O((N^\delta)^t) = O(N^{1-\eps'})$. Thus the total time of the algorithm is $\tilde{O}(N^{1-\eps'}+N^{1- \delta\eps'})$. This contradicts SETH, because SETH implies that no algorithm exists for 2-OV in $O((|U|\cdot|V|)^{1-\eps''}) = O(N^{1-\eps''})$ time for any $\eps''>0$.

\end{proof}

\subsection{$(2-\eps)$-approximation in directed graphs requires quadratic update time}

In this section we give a quadratic lower bound per update for $(2-\eps)$-approximation for  eccentricities and $(2-\eps)$-approximation for radius on directed, unweighted graphs.
\subsubsection{Eccentricities}
\begin{theorem}
\label{thm:2direcc}
Let $t$, $\eps$, and $\eps'$ be positive constants. SETH implies that there exists no fully dynamic algorithm for $(2-\eps)$-approximate all Eccentricities on directed, unweighted graphs with $n$ vertices and $\tilde{O}(n)$ edges, which has preprocessing time $p(n)=O(n^t)$, amortized update time $u(n)=O(n^{2-\eps'})$, and amortized query time $q(n)=O(n^{2-\eps'})$.

The same result holds for the incremental and decremental settings but for worst-case update and query times.
\end{theorem}

\begin{proof}[Proof of Theorem~\ref{thm:2direcc}]$ $

\paragraph{Construction}
\subparagraph{Initialization}

Let $a=\ceil{\frac{3-3\eps}{\eps} }+1.$

We construct a graph $G$ by first creating $G_\delta$ from section \ref{basegraph}. We then replace the paths between nodes $u^a$ and $c_U$ with single directed edges $u^a\rightarrow c_U$, and replace the paths between $c_V$ and $v^0$ with single directed edges $c_V\rightarrow v^0$. Replace undirected edges between $u^i$ and $u^{i+1}$ with $u^i\rightarrow u^{i+1}$, and do the same for the $v^i$ edges. 

We then add two nodes $x$ and $y$. For each node $u^a$, add edge $u^a\rightarrow x$, and for each node $u^0$, add $x\rightarrow u^0$. Finally, add directed edges $c_U\rightarrow y$ for all $c_U\in C_U$, and add a directed path of length $a$ from $y$ to each $c_V\in C_V$. 

\subparagraph{Stages} 

We proceed in $N^{1-2\delta}$ stages, one for each element $w\in W$. For the current $w$, for each coordinate $c$ where $w[c]=1$, we add edge $c_U\rightarrow c_V$. We then query for all eccentricities of $G$. We will show that if there is an orthogonal triple that includes $w$ and a vector $u\in U$, then the eccentricity of $u^a$ in $G$ is $2a+3$; otherwise, the eccentricity of $u^a$ is $a+3$. A $(2-\eps)$-approximation algorithm for eccentricities distinguishes between these two cases and can thus detect an orthogonal triple that includes $w$ if one exists. If such an orthogonal triple is not detected, we undo the edge additions for the stage and continue to the next $w$.

\begin{figure}[h]
  \centering
  \includegraphics[width=.8\linewidth]{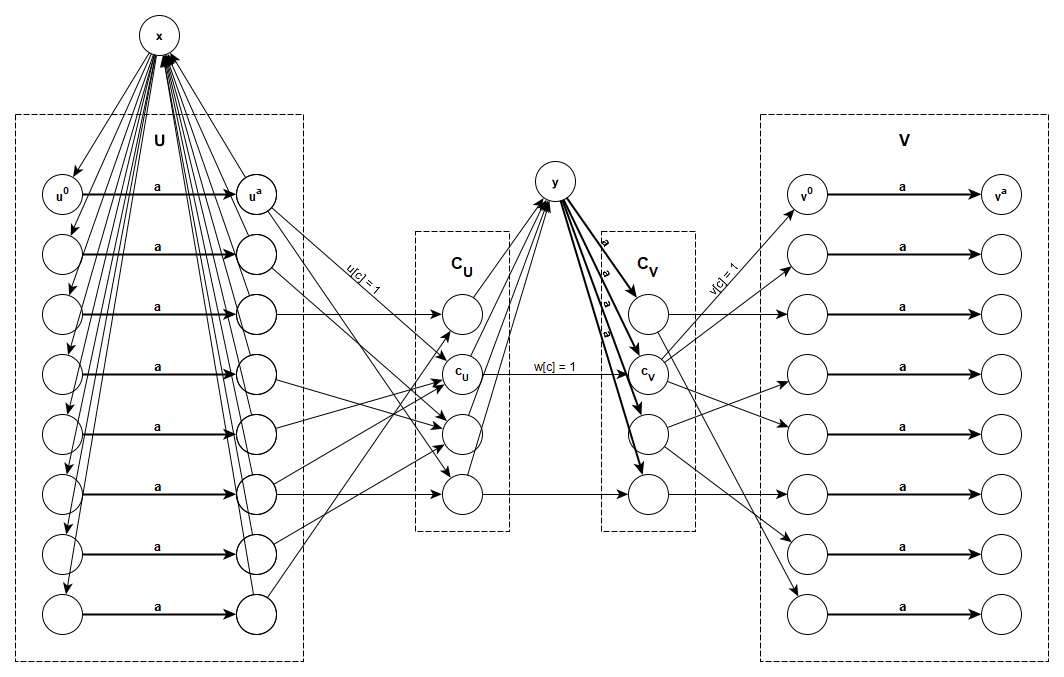}
  \caption{Sketch of Theorem \ref{thm:2direcc} Construction. Bold edges represent directed paths, whose labels denote their length.}
\end{figure}

\paragraph{Analysis}
\label{analysis:2eccdir}

\subparagraph{Correctness}
We first claim that all nodes other than the $v^i$ are at most distance $a+3$ from all nodes $u^a$. The nodes on the paths from $y$ to all $c_V$ (inclusive) are accessible from each $u^a$ via any $c_U$ for which $u[c]=1$ in at most $a+2$ steps, and all nodes $u^i$ are accessible by each $u^a$ in at most $a+2$ steps via $x$. All $u^a$ may thus also reach all $c_U$ in $a+3$ steps via another $u'^a$.

If for the current stage, for all $u\in U$, $v\in V$, $u\cdot v\cdot w\not= 0$, then for each $u,v$ there exists a coordinate $c$ such that $u[c]=v[c]=w[c]=1$. Thus, $d(u^a,v^i) = i + 3 \leq a+3$ for all $u$ and $v$. Thus, the eccentricities of all $u^a$ must be at most $a+3$.

Suppose for the current stage there exists $u\in U$, $v\in V$ such that $u\cdot v\cdot w =0$. Fix $u$ and $v$. We claim that $d(u^a,v^a)\geq 2a+3$. There does not exist a coordinate such that $u[c]=v[c]=w[c]=1$, so any path must either go via $x$ or $y$. We have $d(x,v^a)\geq 2a+4$ and $d(y,v^a)=2a+1$, and we also have $d(u^a,x)=1$ and $d(u^a,y)=2$, so $d(u^a,v^a)=2a+3$.
\end{proof}
\subsubsection{Radius}
\begin{theorem}
\label{thm:2dirrad}
Let $t$, $\eps$, and $\eps'$ be positive constants. The 3-HS hypothesis implies that there exists no fully dynamic algorithm for $(2-\eps)$-approximate Radius on directed, unweighted graphs with $n$ vertices and $\tilde{O}(n)$ edges, which has preprocessing time $p(n)=O(n^t)$, amortized update time $u(n)=O(n^{2-\eps'})$, and amortized query time $q(n)=O(n^{2-\eps'})$.

The same result holds for the incremental and decremental settings but for worst-case update and query times.
\end{theorem}

\begin{proof}[Proof of Theorem~\ref{thm:2dirrad}]$ $

\subsubsection{Construction}
\subsubsection*{Initialization}

Let $a=\ceil{\frac{3-3\eps}{\eps} }+1.$

We construct the graph $G$ in the same way as in Theorem \ref{thm:2direcc}.

\subsubsection*{Stages} 

We proceed in $N^{1-2\delta}$ stages, one for each element $w\in W$. For the current $w$, for each coordinate $c$ where $w[c]=1$, we add edge $c_U\rightarrow c_V$. We then query for the radius of $G$. We will show that if there is a hitting set with $w$ and $u\in U$, then the radius of $G$ is $a+3$; otherwise, the radius is at least $2a+3$. A $(2-\eps)$-approximation algorithm for radius distinguishes between these two cases and can thus detect a hitting set that includes $w$ if one exists. If such a hitting set is not detected, we undo the edge additions for the stage and continue to the next $w$.

\subsubsection{Analysis}

\subsubsection*{Correctness}
All nodes other than the $u^i$ and $x$ have infinite eccentricity, so they cannot be the center of $G$. Each $u^a$ can reach all nodes in $u^i$ within distance $a+2$ and the distance from $u^a$ to every $v^a$ is at least $a+3$, so the farthest node from any $u^a$ must be to the right of it. All paths from $x$ or $u^i$ for $i<a$ to any node $v^i$ must go through $u^a$ for some $u$, and thus must be farther than $u^a$ from the farthest node from $u^a$. Thus, some $u^a$ must be the center.

From section \ref{analysis:2eccdir}, if for the current stage $u$ participates in a hitting set (i.e., has $u[c]=v[c]=w[c]=1$ for all $v$) then $u^a$ has eccentricity $a+3$, and if $u$ participates in an orthogonal triple then $u^a$ has eccentricity at least $2a+3$. Thus, the radius of $G$ is $a+3$ if there is a hitting set for the current stage, and it is $2a+3$ otherwise.
\end{proof}

\subsection{Finite approximation in directed graphs requires linear update time}

In this section we give a linear lower bound for finite approximation of fully dynamic diameter, radius, and fixed-vertex eccentricity for directed, unweighted graphs.

\begin{theorem}
\label{thm:findiam} 
Let $\eps$ be a positive constant. The $k$-Cycle conjecture implies that there exists no fully dynamic algorithm for finite-approximate Diameter, Radius, or fixed-vertex Eccentricity on directed, unweighted graphs with $n$ vertices and $\tilde{O}(n)$ edges, which has preprocessing time $p(n)=O(n^{2-\eps})$, amortized update time $u(n)=O(n^{1-\eps})$, and amortized query time $q(n)=O(n^{1-\eps})$.
\end{theorem}

\begin{proof}[Proof of Theorem~\ref{thm:findiam}]$ $

\subsubsection{Construction}
\subsubsection*{Initialization}

We begin with an instance of $k$-Cycle, with directed graph $G=(V,E)$, where $|V| = n$ and $|E|=\tilde{O}(n)$. We make a new graph $G'=(V',E')$, with two sets of $k+1$ copies of $V$, denoted by $\{V_{left}^i\}_{i=0}^k$ and $\{V_{right}^i\}_{i=0}^k$. $G'$ also has a node $s$ and nodes $t_{left}$ and $t_{right}$. For each directed edge $u\rightarrow v$ in $E$, we add edges $u_{left}^{i}\rightarrow v_{left}^{i+1}$ for all $i$ to $G'$. We do the same on the right side. We also add edges $t_{left}\rightarrow v_{left}$ and $v_{left},\rightarrow s$ to $G'$ for all $v_{left}\in\bigcup_i{V_{left}^i}$, and do the same on the right side.

\begin{figure}[h]
  \centering
  \includegraphics[width=\linewidth]{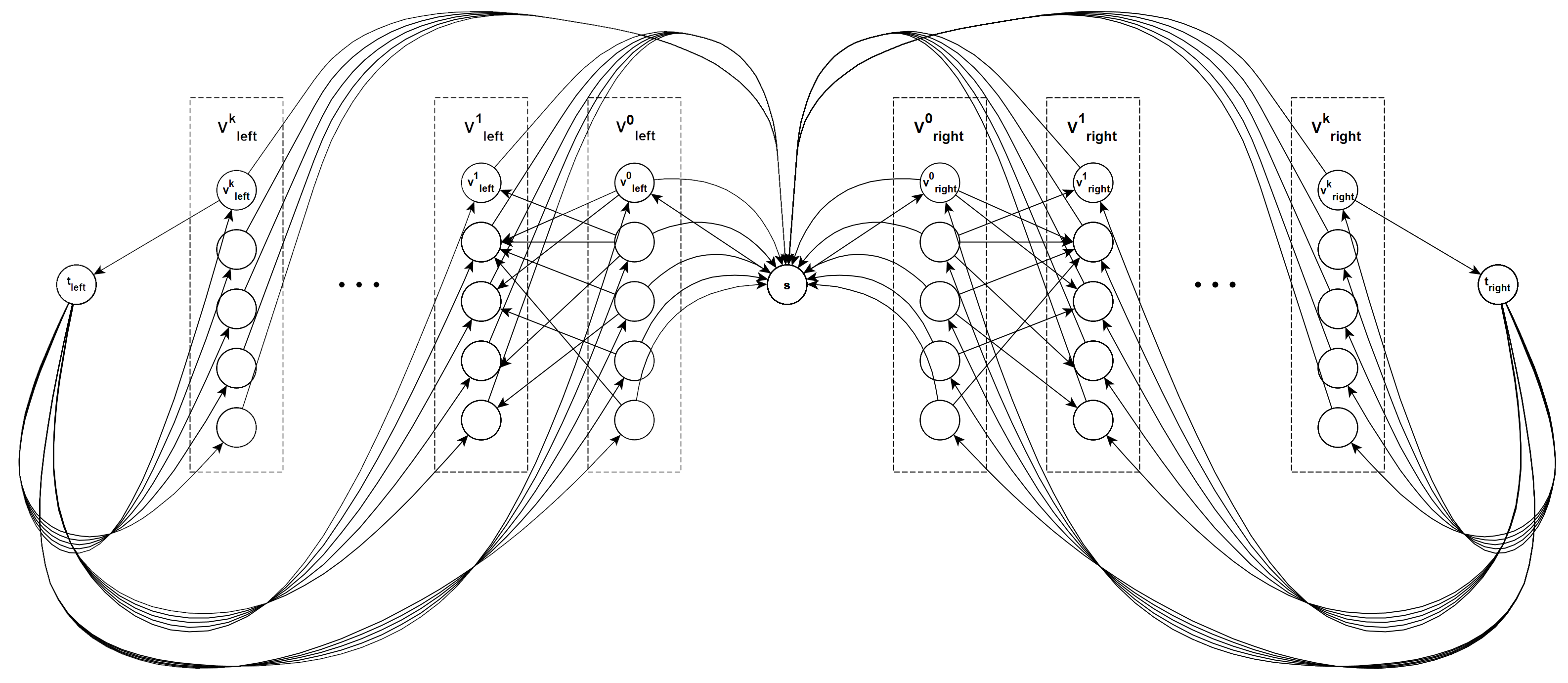}
  \caption{Sketch of Theorem \ref{thm:findiam} Construction.}
\end{figure}

\subsubsection*{Stages} 
We proceed in $n$ stages, one for each element of $V$. For the current $v$, we add edges $s\rightarrow v_{left}^0$, $s\rightarrow v_{right}^0$, $v_{left}^k\rightarrow t$, and $v_{right}^k\rightarrow t$ to $G'$, then query the Diameter or Radius of $G'$, or the Eccentricity of any fixed node in $G'$. We will show that the eccentricities of all nodes in $G'$ are finite if and only if $v$ participates in a $k$-Cycle. If the queried value is infinite, we remove the edges we added in this stage and continue to the next $v$.

\subsubsection{Analysis}

\subsubsection*{Correctness}

We first claim that if for the current stage that there is a $k$-Cycle in $G$ beginning and ending with $v$, then all eccentricities are finite. There is some cycle $p=v, v_1, \dots,v_{k-1},v$ in $G$, where $(v_{i},v_{i+1})\in E$. This means that in $G'$, there is a path $p' = s,v^0, v_1^1,\dots,v_{k-1}^{k-1},v^k,t$ on both the left and right side. By construction, $t_{left}$ can reach all nodes $v_{left}$ and all nodes $v_{left}$ can reach $s$, and the same holds on the right side. Thus, all nodes are at a finite distance from one another via $s$.

We now claim that if for the current stage there is no $k$-Cycle, then all eccentricities are infinite. We note that all nodes on the left side of $s$ must pass through $s$ to reach $t_{right}$, and all nodes on the right side of $s$ must pass through $s$ to reach $t_{left}$. Thus, if $s$ cannot reach $t_{left}$ (or symmetrically, $t_{right}$), then all eccentricities will be infinite. But we have no $k$-Cycle in $G$, so there is no path from $v^0$ to $v^{k}$ on either side of $G'$. The only outgoing edges of $s$ are to $v_{left}^0$ and $v_{right}^0$, so there is no path from $s$ to $t$ on either side. Thus all eccentricities must be infinite.

Distinguishing between a finite and infinite Diameter, Radius, or fixed-vertex Eccentricity is sufficient to distinguish whether all eccentricities are finite or infinite.

\subsubsection*{Running time}
Suppose for the sake of contradiction that the algorithm of Theorem \ref{thm:findiam} exists. The graph $G'$ has $O(n)$ nodes and $\tilde{O}(n)$ edges. We perform $\tilde{O}(n)$ edge updates and $n$ queries, so the total update and query time is $\tilde{O}(n^{2-\eps})$. The preprocessing time is $\tilde{O}(n^{2-\eps})$ as well. This contradicts the $k$-cycle conjecture, which states that no algorithm exists for $k$-cycle in $O(n^{2-f(k)-\eps'})$ time for any $\eps'>0$ and any function $f$ that goes to 0 as $k$ goes to infinity.

\end{proof}

%% file: incremental32.tex
\section{Partially dynamic algorithms}
\subsection{Randomized partially dynamic nearly $3/2$-approximation}


In this section we present a nearly $3/2$-approximation for incremental/decremental diameter, a nearly $3/2$-approximation for incremental/decremental radius, and a nearly $5/3$-approximation for incremental/ decremental eccentricities, given access to a black-box incremental/decremental approximate SSSP algorithm as specified in the preliminaries. 

\subsubsection{Diameter}

\begin{theorem}\label{thm:incred} There is a Las Vegas randomized algorithm for incremental (resp., decremental) diameter in unweighted, directed graphs against an oblivious (resp., adaptive) adversary that given $\eps>0$, runs in total time $\tilde{O}(\max_{D_f\leq D'\leq D_0}\{T_{inc}(n,m,D',\eps)\frac{\sqrt{n/ D'}}{\eps^2}\})$ (resp., $\tilde{O}(\max_{D_0\leq D'\leq D_f}\{T_{dec}(n,m,D',\eps)\frac{\sqrt{n/ D'}}{\eps^2}\})$) with high probability, and maintains an estimate $\hat{D}$ such that $\frac{2(1-\eps)}{3}D-\frac{2}{3}\leq \hat{D}\leq D$ where $D$ is the diameter of the current graph.
\end{theorem}




By Lemma~\ref{lem:max}, the following lemma implies Theorem~\ref{thm:incred}.

\begin{lemma}\label{lem:incred}
There is a Las Vegas randomized algorithm for incremental (resp., decremental) diameter in unweighted, directed graphs that given $ D',\eps>0$, runs in total time $\tilde{O}\left(T_{inc}(n,m,D',\eps)\frac{\sqrt{n/ D'}}{\eps}\right)$ (resp., $\tilde{O}\left(T_{dec}(n,m,D',\eps)\frac{\sqrt{n/ D'}}{\eps}\right)$) with high probability, and maintains an estimate $\hat{D}\leq D$ such that if $ D'\leq D$ then $\hat{D}\geq \frac{2(1-\eps)}{3} D'-\frac{2}{3}$ where $D$ is the diameter of the current graph. The incremental algorithm works against an oblivious adversary.
\end{lemma}

\begin{proof}$ $

\paragraph{Algorithm}
Let $\delta=2\eps/11$. Throughout the incremental (resp., decremental) algorithm we will run in-$\mathcal{A}_{inc}$ (in-$\mathcal{A}_{dec}$) and out-$\mathcal{A}_{inc}$ (out-$\mathcal{A}_{dec}$) from carefully chosen sets of vertices. For ease of notation, we let in-$\mathcal{A}$ denote either in-$\mathcal{A}_{inc}$ or in-$\mathcal{A}_{dec}$, depending on the setting, and similarly we let out-$\mathcal{A}$ denote either out-$\mathcal{A}_{inc}$ or out-$\mathcal{A}_{dec}$. 

\subparagraph{Initialization}
Let $\alpha$ be such that $ D'= \Theta(n^{1-2\alpha})$. We randomly sample a set $S$ of size $\Theta(n^\alpha\log^2 n)$ so that with high probability, for every vertex $v$, $S$ hits $N_{out}(v,n^{1-\alpha})$. For the incremental algorithm, since the adversary is oblivious it is also true that with high probability, for every vertex $v$, \emph{after every update}, $S$ hits $N_{out}(v,n^{1-\alpha})$.

Throughout the entire execution of the algorithm, for all $s\in S$ we run in-$\mathcal{A}(s,D',\delta)$. Additionally, we maintain the approximate distance $d'(v,S)$ from every vertex $v$ to $S$ as described in the preliminaries.
Let $W$ be the dynamically changing set of vertices $v$ that satisfy $d'(v,S)> D'/3$.

\subparagraph{Phases}
The algorithm runs phases. The first phase begins right after initialization. At the beginning of each phase, we choose a vertex $w\in W$ if $W\not=\emptyset$. The decremental algorithm only has one phase and we let $w$ be an arbitrary vertex in $W$. We note that in the decremental setting, distances can only increase so $w$ never leaves $W$. 

In the incremental setting on the other hand, distances can decrease so vertices can leave $W$. The incremental algorithm may have many phases, and at the beginning of each phase, we choose $w\in W$ uniformly at random. The beginning of a new phase is triggered when $w$ leaves the set $W$.

Throughout the phase, we run out-$\mathcal{A}(w,D',\delta)$. Also, we will define a subset $S'\subseteq B_{out}(w, \frac{ D'}{3})$ and for all $s'\in S'$, we run in-$\mathcal{A}(s',D',\delta)$. $S'$ is initially empty and we independently add each vertex in $B_{out}(w, \frac{ D'}{3})$ to $S'$ with probability $\min\{1,\frac{\log^2 n}{\delta  D'}\}$. In the incremental setting (but not the decremental setting), $B_{out}(w, \frac{ D'}{3})$ can grow, and whenever a vertex $u$ joins $B_{out}(w, \frac{ D'}{3})$ we add $u$ to $S'$ with probability $\min\{1,\frac{\log^2 n}{\delta  D'}\}$.

\subparagraph{Reinitialization}

If at any point during the execution of the algorithm, any of the events listed below occur, we reinitialize the entire algorithm. We will show in the analysis that with high probability we never reinitialize the algorithm.

\begin{itemize}
    \item $|B_{out}(w, \frac{ D'}{3})|> n^{1-\alpha}$
    \item $|S'|> \frac{n^{1-\alpha}\log^4 n}{\delta  D'}$
    \item there is a vertex $v\in B_{out}(w, \frac{ D'}{3})$ such that $d'(v,S')>\delta  D'$.
\end{itemize}

\subparagraph{Query}

Following each update, the return value $\hat{D}$ is the maximum distance estimate found over all instantiations of out-$\mathcal{A}$ and in-$\mathcal{A}$. That is, $\hat{D}=\max\{\max_{v\in V}d'(w,v),\max_{v\in V,s\in S\cup S'}d'(v,s)\}$. 

To maintain this value, we maintain the following heaps. For every vertex $v$ that we run out-$\mathcal{A}$ (resp., out-$\mathcal{A}$) from, we keep a max-heap ${\cal H}(v)$ that stores for each 
other vertex $u$ the estimate $d'(v,u)$ (resp., $d'(u,v)$). Let $\hat d_{out}(v)$ be the value that ${\cal H}(v)$ outputs.
Additionally we keep a max-heap $\mathcal{H}$ which stores each
$\hat d_{out}(v)$.

\paragraph{Analysis}

\subparagraph{Correctness}

The return value $ \hat{D}$ is $d'(u,v)$ for some $u,v$, so $ \hat{D}\leq D$. It remains to show that if $D'\leq D$, then $\frac{2(1-\eps)}{3}D'-\frac{2}{3}\leq  \hat{D}$.


Let $s^*$ and $t^*$ be the true diameter endpoints.

\noindent{\bf Case 1: $s^*\not\in W$.} Let $s\in S$ be such that $d'(s^*,s)\leq \frac{ D'}{3}$. Then $d(s^*,s)\leq D'(\frac{1}{3(1-\delta)})$. Then by the triangle inequality, $d(s,t^*)\geq D- D'(\frac{1}{3(1-\delta)})\geq  D'(1-\frac{1}{3(1-\delta)})$. Thus, $\hat{D}\geq d'(s,t^*)\geq  D'(1-\delta)(1-\frac{1}{3(1-\delta)})= D'(\frac{2}{3}-\delta$).

\noindent{\bf Case 2: $s^*\in W$.} We don't explicitly use the fact that $s^*\in W$; we just use the fact that $w$ exists. If $d'(w,t^*)\geq \frac{2 D'}{3}$, then we are done. So suppose otherwise; that is, suppose $d'(w,t^*)< \frac{2 D'}{3}$ so $d(w,t^*)< D'( \frac{2}{3(1-\delta)})$. 

Consider the shortest path from $w$ to $t^*$. Let $q$ be the vertex on this path at distance $\lfloor\frac{ D'}{3}\rfloor$ from $w$. So $d(q,t^*)<  D'(\frac{2}{3(1-\delta)})-\lfloor\frac{ D'}{3}\rfloor\leq  D'(\frac{2}{3(1-\delta)}-\frac{1}{3})+\frac{2}{3}$. We know that $q\in B_{out}(w, \frac{ D'}{3})$, so $d'(S',q)\leq\delta D'$ or else we would have reinitialized the algorithm. Thus, by Claim~\ref{claim:pre} from the preliminaries, $d(S',q)\leq\frac{\delta D'}{1-2\delta}$. Let $q'\in S'$ be a vertex with $d(q',q)\leq\frac{\delta D'}{1-2\delta}$.

By the triangle inequality, $d(q',t^*)\leq d(q',q)+d(q,t^*)\leq  D'(\frac{\delta}{1-2\delta}+\frac{2}{3(1-\delta)}-\frac{1}{3})+\frac{2}{3}$. By the triangle inequality, $d(s^*,q')\geq D- D'(\frac{\delta}{1-2\delta}+\frac{2}{3(1-\delta)}-\frac{1}{3})-\frac{2}{3}\geq  D'(\frac{4}{3}-\frac{\delta}{1-2\delta}-\frac{2}{3(1-\delta)})-\frac{2}{3}$. Thus, $\hat{D}\geq d'(s^*,q')\geq D'(1-\delta)(\frac{4}{3}-\frac{\delta}{1-2\delta}-\frac{2}{3(1-\delta)})-\frac{2}{3})= D'(\frac{4(1-\delta)}{3}-\frac{\delta(1-\delta)}{1-2\delta}-\frac{2}{3})-\frac{2}{3}$.


Setting $\delta= 2\eps/11$ completes the proof of correctness.



\subparagraph{Running time}

\subparagraph{Reinitialization}
We will argue that with high probability, we never reinitialize the algorithm. One event that triggers algorithm reinitialization is if $|B_{out}(w, \frac{ D'}{3})|> n^{1-\alpha}$. 
Recall that with high probability, for every vertex $v$, $S$ hits $N_{out}(v,n^{1-\alpha})$; this is true for the decremental algorithm only initially, and for the incremental after every update.
By the definition of $W$, $ D'/3<d'(w,S)\leq d(w,S)$. Thus, $B_{out}(w, \frac{ D'}{3})$ contains no vertices in $S$. So, with high probability $B_{out}(w, \frac{ D'}{3})$ does not contain $N_{out}(v,n^{1-\alpha})$. That is, with high probability, $|B_{out}(w, \frac{ D'}{3})|< n^{1-\alpha}$. For the decremental algorithm we have shown that this inequality holds only for the initial graph, however it also holds after every update since $|B_{out}(w, \frac{ D'}{3})|$ can only decrease over time. Thus, for both the incremental and decremental algorithms, with high probability we never reinitialize the algorithm due to $|B_{out}(w, \frac{ D'}{3})|> n^{1-\alpha}$.

Another event that triggers reinitialization is if
$|S'|> \frac{n^{1-\alpha}\log^2 n}{\delta  D'}$. $|S'|$ is a random variable drawn from a binomial distribution with  $p=\frac{\log^2 n}{\delta D'}$ and expected value $\frac{|B_{out}(w, \frac{ D'}{3})|\log^2 n}{\delta D'}$. We know that $|B_{out}(w, \frac{ D'}{3})|\leq n^{1-\alpha}$ or else we would have reinitialized the algorithm due to the above event. Thus, $|S'|\leq  \frac{n^{1-\alpha}\log^4 n}{\delta  D'}$ with high probability.

The last event that triggers reinitialization is if there is a vertex $v\in B_{out}(w, \frac{ D'}{3})$ such that $d'(v,S')>\delta  D'$.
The expected size of $S'$ is $\frac{|B_{out}(w, \frac{ D'}{3})|\log^2 n}{\delta D'}$. Thus, with high probability, for all vertices $v\in B_{out}(w, \frac{ D'}{3})$, after every update, $S'$ hits a vertex of distance at most $\delta D'$ to $v$.

We have shown that each of the three events that trigger reinitialization do not occur with high probability (probability at least $1-1/n^c$ for all constants $c$). Thus, with high probability we never reinitialize the algorithm.

\subparagraph{Running in-$\mathcal{A}$ and out-$\mathcal{A}$}
We will calculate the total number $a$ of vertices that we ever run in-$\mathcal{A}$ or out-$\mathcal{A}$ from. Then the total time is $\tilde{O}(aT_{inc}(n,m,D',\eps))$; maintaining the heaps ${\cal H}(v)$ and $\mathcal H$ increases the running time only by a
factor of $O(\log n)$.

In the initialization step, we initialize in-$\mathcal{A}$ from all $\tilde{O}(n^{\alpha})$ vertices in $S$ as well as a dummy vertex. Throughout each phase, we run out-$\mathcal{A}$ from $w$ and we run in-$\mathcal{A}$ from all $\tilde{O}(\frac{n^{1-\alpha}}{\eps D'})$ vertices that are added to $S'$ during the phase, as well as a dummy vertex.

The decremental algorithm has only one phase. We calculate the number of phases in the incremental algorithm. The beginning of a new phase is triggered when $w$ leaves the set $W$. 
We note that since we are in the incremental setting, no vertices are added to $W$ during a phase. The sequence of updates dictates if and when each vertex is removed from $W$. Each update may trigger any number of vertices to leave $W$.


Fix a choice of $w$ and let $W_0$ be the set $W$ at the point in time that the algorithm chooses $w$. We say that $w$ is a \emph{success} if at most half of the vertices in $W_0$ leave $W$ after $w$ leaves $W$. Since $w$ is chosen randomly and the adversary is oblivious, the probability that $w$ is a success is at least $1/2$. Once $\log_2 n$ choices of $w$ are successful, then $W$ is empty. Let $Y$ be a random variable defined as the number of times the algorithm chooses a new $w$ until $W$ is empty. $Y$ is a negative binomial random variable, which implies the following concentration bound for any constant $c$: $P[Y > 2c\log_2 n]\leq \exp(\frac{-c(1-1/c)^2}{2}\log_2 n)$. Thus, $Y=\tilde{O}(1)$ with high probability. 


Putting everything together, with high probability, $a=\tilde{O}(n^{\alpha}+\frac{n^{1-\alpha}}{ \eps D'})=\tilde{O}\left(\sqrt{n/ D'}\right)$, so the total time is $\tilde{O}\left(T_{inc}(n,m,D',\eps)\frac{\sqrt{n/ D'}}{\eps}\right)$.
\end{proof}

\subsubsection{Radius}
Our theorems and proofs for radius are analogous to those for diameter, with a number of key differences. In this section we describe the differences. 

\begin{theorem}\label{thm:incredrad} There is a Las Vegas randomized algorithm for incremental (resp., decremental) radius in unweighted, undirected graphs that given $\eps>0$ runs in total time $\tilde{O}(\max_{R_f\leq R'\leq R_0}\{T_{inc}(n,m,2R',\eps)\frac{\sqrt{n/ R'}}{\eps^2}\})$ (resp., $\tilde{O}(\max_{R_0\leq R'\leq R_f}\{T_{dec}(n,m,2R',\eps)\frac{\sqrt{n/ R'}}{\eps^2}\})$) with high probability, and maintains an estimate $\hat{R}$ such that $R\leq \hat{R}\leq (1+\eps)(\frac{3}{2}R+\frac{1}{2})$ where $R$ is the radius of the current graph. The incremental algorithm works against an oblivious adversary.
\end{theorem}

By Lemma~\ref{lem:min}, the following lemma implies Theorem~\ref{thm:incredrad}.

\begin{lemma}\label{lem:incredrad}
There is a Las Vegas randomized algorithm for incremental (resp., decremental) radius in unweighted, undirected graphs that given $ R'$ and $\eps>0$, runs in total time $\tilde{O}\left(T_{inc}(n,m,2R',\eps)\frac{\sqrt{n/ R'}}{\eps}\right)$ (resp., $\tilde{O}\left(T_{dec}(n,m,2R',\eps)\frac{\sqrt{n/ R'}}{\eps}\right)$) with high probability, and maintains an estimate $\hat{R}\geq R$ such that if $ R'\geq R$ then $\hat{R}\leq (1+\eps)(\frac{3}{2}R'+\frac{1}{2})$ where $R$ is the radius of the current graph. The incremental algorithm works against an oblivious adversary.
\end{lemma}

\begin{proof}$ $
\paragraph{Algorithm}
The algorithm is identical to the algorithm of Lemma~\ref{lem:incred} except for the following substitutions:
\begin{itemize}
 \item Set $\delta=\eps/4$.
    \item All instantiations of in-$\mathcal{A}$ and out-$\mathcal{A}$ are replaced by the corresponding algorithm $\mathcal{A}$ for undirected graphs and the input parameter $k$ is set to $2R'$ instead of $D'$.
    \item All instances of $D'/3$ are replaced by $R'/2$.
    \item All remaining instances of $D'$ are replaced by $R'$.
    \item The return value $\hat{R}$ is the minimum eccentricity estimate found over all instantiations of $\mathcal{A}$, with a correction factor of $\frac{1}{1-\delta}$. That is, $\hat{R}=\frac{1}{1-\delta}(\min_{s\in S\cup S'\cup \{w\}}\max_{v\in V}d'(v,s))$. To maintain this value, we maintain the following heaps. For every vertex $v$ that we run $\mathcal{A}$ from, we keep a max-heap ${\cal H}(v)$ that stores for each 
other vertex $u$ the estimate $d'(v,u)$. Let $\hat d_{out}(v)$ be the output value of ${\cal H}(v)$.
Additionally we keep a min-heap $\mathcal{H}$ which stores each
$\hat d_{out}(v)$.
    
\end{itemize}
\paragraph{Analysis}
The running time analysis is identical to that of Lemma~\ref{lem:incred} with the above substitutions. We present the analysis of correctness.

Let $c^*$ be a vertex such that $\varepsilon(c^*)=R$. For all $s\in S\cup S'\cup \{w\}$, let $\varepsilon'(s)$ be the estimated eccentricity $\max_{v\in V}d'(s,v)$.

The return value $ \hat{R}$ is $\frac{1}{1-\delta}(\varepsilon'(v))\geq \varepsilon(v)$ for some $v$, so $\hat{R}\geq R$. It remains to show that if $ R'\geq R$, then $\hat{R}\leq (1+\eps)(\frac{3}{2}R'+\frac{1}{2})$.

\noindent{\bf Case 1: $c^*\not\in W$.} Let $s\in S$ be such that $d'(c^*,s)\leq \frac{ R'}{2}$. Then $d(c^*,s)\leq R'(\frac{1}{2(1-\delta)})$. Then by the triangle inequality, $\varepsilon(s)\leq R'(\frac{1}{2(1-\delta)})+R\leq  R'(\frac{1}{2(1-\delta)}+1)$, so $\varepsilon'(s)\leq R'(\frac{1}{2(1-\delta)}+1)$. Thus, $\hat{R}\leq \frac{1}{1-\delta}(\varepsilon'(s))\leq R'(\frac{1}{1-\delta})(\frac{1}{2(1-\delta)}+1)$

\noindent{\bf Case 2: $c^*\in W$.} We don't explicitly use the fact that $c^*\in W$; we just use the fact that $w$ exists.
Consider the shortest path from $w$ to $c^*$. Let $q$ be the vertex on this path at distance $\lfloor\frac{ R'}{2}\rfloor$ from $w$. So $d(q,c^*)\leq \lceil\frac{R'}{2}\rceil\leq \frac{R'}{2}+\frac{1}{2}$. We know that $q\in B_{out}(w, \frac{ R'}{2})$, so $d'(S',q)\leq\delta R'$ or else we would have reinitialized the algorithm. Thus, by Claim~\ref{claim:pre} in the preliminaries, $d(S',q)\leq\frac{\delta R'}{1-2\delta}$. Let $q'\in S'$ be a vertex with $d(q',q)\leq\frac{\delta R'}{1-2\delta}$.

By the triangle inequality, $d(q',c^*)\leq d(q',q)+d(q,c^*)\leq  R'(\frac{\delta }{1-2\delta}+\frac{1}{2})+\frac{1}{2}$. By the triangle inequality, $\varepsilon(q')\leq R'(\frac{\delta }{1-2\delta}+\frac{1}{2})+\frac{1}{2}+R\leq R'(\frac{\delta }{1-2\delta}+\frac{3}{2})+\frac{1}{2}$, so $\varepsilon'(q')\leq R'(\frac{\delta }{1-2\delta}+\frac{3}{2})+\frac{1}{2}$. Thus, $\hat{R}\leq (\frac{1}{1-\delta})\varepsilon'(q')\leq (\frac{1}{1-\delta})(R'(\frac{\delta }{1-2\delta}+\frac{3}{2})+\frac{1}{2})$.

Setting $\delta= \eps/4$ completes the proof of correctness.

\end{proof}

\subsubsection{Eccentricities}
Our theorems and proofs for eccentricities are analogous to those for diameter, with a number of key differences. In this section we describe the differences.

\begin{theorem}\label{thm:incredecc} There is a Las Vegas randomized algorithm for incremental (resp., decremental) eccentricities in unweighted, undirected graphs that given $\eps$ with $0<\eps<0.45$ runs in total time \\$\tilde{O}(\max_{R_f\leq D'\leq D_0}\{T_{inc}(n,m,D',\eps)\frac{\sqrt{n/ D'}}{\eps^2}\})$ (resp., $\tilde{O}(\max_{R_0\leq D'\leq D_f}\{T_{dec}(n,m,D',\eps)\frac{\sqrt{n/ D'}}{\eps^2}\})$) with high probability, and for all $v\in V$ maintains an estimate $\hat{\varepsilon}(v)$ such that $\frac{3(1-\eps)}{5}\varepsilon(v)-1\leq \hat{\varepsilon}(v)\leq \varepsilon(v)$ where $\varepsilon(v)$ is the eccentricity of $v$ in the current graph. The incremental algorithm works against an oblivious adversary.
\end{theorem}

By Lemma~\ref{lem:max}, the following lemma implies Theorem~\ref{thm:incredrad}.

\begin{lemma}\label{lem:incred}
There is a Las Vegas randomized algorithm for incremental (resp., decremental) eccentricities in unweighted, directed graphs that given $ D',\eps>0$, runs in total time $\tilde{O}\left(T_{inc}(n,m,D',\eps)\frac{\sqrt{n/ D'}}{\eps}\right)$ (resp., $\tilde{O}\left(T_{dec}(n,m,D',\eps)\frac{\sqrt{n/ D'}}{\eps}\right)$) with high probability, and for all $v\in V$ maintains an estimate $\hat{\varepsilon}(v)\leq \varepsilon(v)$ such that if $D'\leq \varepsilon(v)$, then $\hat{\varepsilon}(v) \geq \frac{3(1-\eps)}{5}\varepsilon(v)-1$. The incremental algorithm works against an oblivious adversary.
\end{lemma}

\begin{proof}$ $


\paragraph{Algorithm}


We run the algorithm from Lemma~\ref{lem:incred} with the following substitutions:
\begin{itemize}
 \item Set $\delta=\eps/9$.
     \item All instantiations of in-$\mathcal{A}$ and out-$\mathcal{A}$ are replaced by the corresponding algorithm $\mathcal{A}$ for undirected graphs.
    \item All instances of $D'/3$ are replaced by $2D'/5$.
    \item For all $s\in S\cup S'\cup \{w\}$, let $\varepsilon'(s)$ be the estimated eccentricity $\max_{v\in V}d'(s,v)$. For all $v\in V$, the return value is $\hat{\varepsilon}(v)=\max_{s\in S\cup S'\cup \{w\}} \{d'(v,s),(1-\delta)\varepsilon'(s)-\frac{d'(v,s)}{1-\delta}\}$. To maintain these values, for every vertex $v$ we keep a max-heap $\mathcal{H}(v)$ which stores $d'(v,s)$ and $(1-\delta)\varepsilon'(s)-\frac{d'(v,s)}{1-\delta}$ for each vertex $s$ that we run $\mathcal{A}$ from.
\end{itemize}

\paragraph{Analysis}
The running time analysis is identical to that of Lemma~\ref{lem:incred} with the above substitutions. 
We present the analysis of correctness. 
For all $v\in V$, let $v'$ be the farthest vertex from $v$. Fix $v\in V$. 

First we show that $\hat{\varepsilon}(v)\leq {\varepsilon}(v)$. If $\hat{\varepsilon}(v)=d'(v,s)$ for some $s$ then it is clear that $\hat{\varepsilon}(v)\leq \varepsilon(v)$. Otherwise, from the expression of the return value, for some $s$, $\hat{\varepsilon}(v)\leq(1-\delta)\varepsilon'(s)-\frac{d'(v,s)}{1-\delta}\leq \varepsilon(s)-d(v,s)\leq d(v,s')\leq \varepsilon(v)$. It remains to show that if $D'\leq \varepsilon(v)$, then $\frac{3(1-\eps)}{5}\varepsilon(v)-1\leq \hat{\varepsilon}(v)$.




\noindent{\bf Case 1. $v'\not\in W$.} 
Let $s\in S$ be such that $d'(v',s)\leq 2D'/5$. Then $d(v',s)\leq\frac{D'}{5(1-\delta)}$. By the triangle inequality, $d(v,s)\geq \varepsilon(v)-d(v',s)\geq \varepsilon(v)-\frac{2D'}{5(1-\delta)}\geq \varepsilon(v)(1-\frac{2}{5(1-\delta)})$. Then, $\varepsilon'(v)\geq d'(v,s)\geq \varepsilon(v)(1-\delta)(1-\frac{2}{5(1-\delta)})=\varepsilon(v)(3/5-\delta)$.

\noindent{\bf Case 2. $v'\in W$.} We don't explicitly use the fact that $v'\in W$; we just use the fact that $w$ exists. If $d'(w,v)\geq\frac{3\varepsilon(v)}{5}$, then we are done. So suppose otherwise; that is, suppose $d'(w,v)<\frac{3\varepsilon(v)}{5}$ so $d(w,v)<\frac{3\varepsilon(v)}{5(1-\delta)}$. Consider the shortest path from $w$ to $v$. Let $q$ be the vertex on this path at distance $\lfloor\frac{2D'}{5}\rfloor$ from $w$. So $d(q,v)<\frac{3\varepsilon(v)}{5(1-\delta)}-\lfloor\frac{2D'}{5}\rfloor\leq \frac{3\varepsilon(v)}{5(1-\delta)}-\frac{2D'}{5}+\frac{4}{5}$. We know that $q\in B_{out}(w,\frac{2D'}{5})$, so $d'(S',q)\leq \delta D'$ or else we would have reinitialized the algorithm. Thus, by Claim~\ref{claim:pre} in the preliminaries, $d(S',q)\leq \frac{\delta D'}{1-2\delta}$. Let $x\in S'$ be a vertex with $d(q,x)\leq \frac{\delta D'}{1-2\delta}$.

Recall from the description of the return value that $\hat{\varepsilon}(v)\geq (1-\delta)\varepsilon'(x)-\frac{d'(v,x)}{1-\delta}\geq(1-\delta)^2\varepsilon(x)-\frac{d(v,x)}{1-\delta}$. 
By the triangle inequality, $\varepsilon(x)\geq d(x,v')\geq \varepsilon(v)-d(v,x)$. So,  $\hat{\varepsilon}(v)\geq (1-\delta)^2(\varepsilon(v)-d(v,x))-\frac{d(v,x)}{1-\delta}\geq (1-\delta)^2\varepsilon(v)-d(v,x)((1-\delta)^2+\frac{1}{1-\delta})$.

By the triangle inequality, $d(x,v)\leq d(x,q)+d(q,v)\leq \frac{\delta D'}{1-2\delta}+\frac{3\varepsilon(v)}{5(1-\delta)}-\frac{2D'}{5}+\frac{4}{5}\leq \varepsilon(v)(\frac{\delta}{1-2\delta}+\frac{3}{5(1-\delta)}-\frac{2}{5})+\frac{4}{5}$. 

Combining the previous two paragraphs, we have $\hat{\varepsilon}(v)\geq\varepsilon(v)((1-\delta)^2- (\frac{\delta}{1-2\delta}+\frac{3}{5(1-\delta)}-\frac{2}{5})((1-\delta)^2+\frac{1}{1-\delta}))-1$.

One can verify that for all $\delta<0.05$, $(1-\delta)^2- (\frac{\delta}{1-2\delta}+\frac{3}{5(1-\delta)}-\frac{2}{5})((1-\delta)^2+\frac{1}{1-\delta})\geq 3(1-9\delta)/5$. Thus we can set $\delta=\eps/9$ and obtain the result.

\end{proof}

%% file: incremental1eps.tex
\subsection{Deterministic incremental $(1+\eps)$-approximation}\label{subsec:det}

In this section we present a {\em deterministic} incremental $(1+\eps)$-approximation
for diameter, radius, and eccentricities in directed graphs. The algorithm for eccentricities only works for strongly connected graphs.
\subsubsection{Diameter}

\begin{theorem}\label{thm:incred2} 
There is a deterministic algorithm  for incremental diameter in unweighted, directed  graphs that, for any   $\eps$ with  $0<\eps<2$, runs in total time $\tilde{O}(\max_{D_f\leq D'\leq D_0} \{(T_{inc}(n,m,D',\eps) +
m) n/(\eps^2 D') \})$, and maintains an estimate $\hat{D}$ such that $(1-\eps)D\leq \hat{D}\leq D$, where $D$ is the diameter of the current graph.
\end{theorem}


By Lemma~\ref{lem:max}, the following lemma implies Theorem~\ref{thm:incred2}.

\begin{lemma}\label{lem:incred2}
There is a deterministic algorithm for incremental diameter in unweighted, directed graphs that,  given a parameter $D'$, for any $\eps$ with $0<\eps<2$, runs in total time $\tilde{O}((T_{inc}(n,m,D',\eps) + m) n/{(\eps D')})$, and maintains an estimate $\hat{D} \le D$ such that if $D' \le D$ then  $\hat{D}\geq (1-\eps)D'$, where $D$ is the diameter of the current graph.
\end{lemma}

\begin{proof}$  $

The basic idea for the proof is that it is possible deterministically to pick a set of {\em center} nodes, to estimate the in-eccentricity of each of the centers (i.e. for each center $c$, $\max_{v\in V}d(v,c)$) by running in-$\mathcal{A}$ 
and to return the maximum values of these estimates. Let $s^*$ and $t^*$ be vertices such that $d(s^*, t^*) = D$. Let $I$ be the set of nodes
 that can reach $t^*$ by a path of length at most $\eps D$. Then  for every vertex
$x \in I$ it must hold that $d(s^*,x)\geq (1-\eps)D$. Thus, $x$ has in-eccentricity at least $(1-\eps) D$. Hence, we will chose centers in such a way that
it is guaranteed that one of the nodes in $I$ is a center. Since we do not know $t^*$, we pick a set of centers such that for {\em every} node we can guarantee
that it has distance at most $\eps D$ from a center.
The proof below uses this basic idea, but it slightly more complicated as the SSSP data structure that we use as subroutine only gives approximate distances.

\paragraph{Algorithm}
Let $\eps'=\eps/2$. When we run in-$\mathcal{A}_{inc}$ from a vertex $v$, we also keep a max-heap ${\cal H}(v)$ that stores for each 
other vertex $u$ the estimate $d'(u,v)$  returned by the algorithm and which outputs the value $\hat d_{in}(v)
=\max_u d'(u,v)$. Additionally we keep a max-heap $\mathcal{H}$ which stores for each center $v$ the value
$\hat d_{in}(v)$.
This heap is updated whenever $\hat d_{in}(v)$
changes for any center $v$.


Our deterministic algorithm is a derandomized version of the following randomized Monte-Carlo algorithm: 
Pick a set $S$ of $C n \log n /(\eps' D')$  random nodes (for a suitable constant C), which 
are called {\em centers} and for each center $c$, run
in-$\mathcal{A}_{inc}(c,D',\eps')$. 
Output as $\hat{D}$ the maximum over all nodes $v \in S$ of $\hat d_{in}(v)$.

For the correctness of this algorithm we have to show that
(i) $\hat{D} \le D$ and (ii) if $D' \le D$ then $\hat{D} \ge (1-\eps) D'$.
(i) $\hat{D} \le D$  follows from the fact that every value $\hat  d_{in}(v)$ represents the length
of an actual path in $G$.
(ii) By a standard Chernoff bound one can show that whp for every pair of
vertices $(u, w)$ there exists a center $v$ that lies on the shortest path from
$u$ to $w$ such that $d(v,w) \le \eps' D'$. This implies that whp there exists a
center $v^*$ that lies on the shortest path from $s^*$ to $t^*$ such that
$d(v^*,t^*) \le \eps' D'$, 
where $s^*$ and $t^*$ are vertices with
$d(s^*,t^*) = D$.  It follows that whp $d(s^*,v^*) \ge D - \eps' D'.$
Note that $D' \le D$ implies that whp $d(s^*,v^*) \ge D'(1-\eps')$.
Since we run in-$\mathcal{A}_{inc}$ from $v^*$ and $\hat{D} \ge \hat{d}_{in}(v^*)$
it follows that whp $\hat{D} \ge (1-\eps) D'$, which is the bound we needed to show.
We run
$O(n/(\eps' D')) = O(n/(\eps D'))$ instantiations of algorithm in-$\mathcal{A}_{inc}$, which gives
a 
total time of $\tilde{O}(T_{inc}(n,m,D',\eps) n/{(\eps D')})$.

We now show how to pick centers in a deterministic way such that for every vertex $u$, there exists a center $v$ such that $d(v,u)\leq \eps'D'$, which derandomizes the above algorithm. Our construction is similar to the deterministic
center construction in~\cite{HenzingerKNSICOMP16}.

Let $Out(v, \gamma, H) = \{x \in H, d_H(v,x) \le \gamma\}$ be the set of nodes that vertex $v$ can reach in at most $\gamma$ steps in $H$.

\subparagraph{Initialization}
We construct a set $S $ of vertices, called {\em centers}, as follows:
\begin{enumerate}
\item Set $G' = G$. 
\item Repeatedly add to $S$ a node $v$ of $G'$ with $|Out(v, \eps' D'/2, G')| \ge
\lfloor\eps' D'/2\rfloor$ and remove from $G'$ all nodes 
$x$ in $Out(v, \eps' D'/2, G')$, labeling each of them with their center $v$.
The procedure stops when for every vertex $v$ in $G'$ it holds that  
$|Out(v, \eps' D'/2, G')| < \lfloor\eps' D'/2\rfloor.$
\item Determine the set $A$ of  source nodes $u$ of $G'$ (i.e. vertices with indegree 0) and for each source $u$ store the center label of one of the in-neighbors $w$ of $u$ in $G$ as value
$c(u)$. (Since $G$ is strongly connected, $u$ must have at least one in-neighbor in $G$ and since
$u$ is a source in $G'$, this in-neighbor must have been removed from $G'$ and, thus, be labeled.)
\item Repeatedly remove a node $u$ from $A$, label all nodes in $Out(u, n, 
G')$ with $c(u)$ and remove them from $G'$. 
The procedure stops when $A$ is empty.
\end{enumerate}

\subparagraph{Edge insertion}
We run the algorithm in-$\mathcal{A}_{inc}$
for every center and update them and the corresponding heaps after each edge insertion in $G$.


\subparagraph{Query}
The algorithm returns the maximum value in $\mathcal{H}$.

\paragraph{Analysis}

\subparagraph{Correctness}
We first show that for every vertex $u$ of $G$ there exists a center $v$ such that 
$d(v,u) \le \eps' D'$, i.e. $v$ has a path of length at most $\eps' D'$ to $u$. 

\begin{claim}\label{claim:near2}
After every edge update there exists for every vertex $u$ of $G$  a center  $v$ such that  $d(v,u) \le \eps' D'$.
\end{claim}
\begin{proof}
We first argue that the claim holds after the initialization. 
For the vertices $x$ that receive their center label in Step 2 the claim holds as all these vertices belong to the set
$Out(v, \eps' D'/2, G')$ of some center $v$ and $G'$ is a subgraph of $G$.
Now all nodes that are not labeled in Step 2, belong to $G'$ in Step 3
and, thus,  must belong to
$Out(u, n, G')$ for some source node $u$ in $A$. Hence, it will be labeled  by a center during
Step 4. We now show that this label also fulfills the condition of the claim.

For every source $u$ in $A$ in Step 3 of the initialization it holds that
$|Out(u, \eps' D'/2, G')| < \lfloor\eps' D'/2\rfloor$ which implies that for all vertices
$x$ in $G'$ either $d_{G'}(u,x) < \lfloor\eps' D'/2\rfloor$ or $d_{G'}(u,x) = \infty$.
This implies that in Step 3 of the initialization $Out(u, n, G') = Out(u, \eps' D'/2 -1, G')$. Let $v$ be the value of $c(u)$ as set in Step 3. 
It follows that  $u$ has an edge from a vertex $w$ that received $v$ as center label in Step 2. 
With the previous argument this implies that $d_G(v,u) \le d_G(v,w) + 1 \le \eps' D'/2 +1$ and, thus, that for all nodes $z$ in $Out(u, \eps' D'/2 -1, G')$ it holds that
$d_G(v,z) \le  d_G(v,u) + d_{G'}(u,z) \le \eps' D'/2 +1 + \eps' D'/2 -1=
\eps' D'$.

Thus, the claim holds after initialization. 
Note that edge insertions cannot decrease distances between nodes and, thus, the
claim continues to hold throughout the algorithm.
\end{proof}

The remainder of the correctness argument is identical to that of the randomized algorithm that we have derandomized.

\subparagraph{Running time}

\begin{claim}\label{claim:centernum2}
There are $O({\frac{n}{\eps D'}})$ centers.
\end{claim}
\begin{proof}
Centers are only created in Step 2 of the initialization. When a vertex $v$ becomes a
center, it holds that $|Out(v, \eps' D'/2, G')| \ge
\eps' D'/2$. Assign the nodes that are in $Out(v, \eps' D'/2, G')$  at this 
point in time to $v$. Note that these nodes are not assigned to any other center as they
are removed from $G'$ as soon as they are assigned to $v$. Thus we uniquely assigned to each center at least
$\eps' D'/2 = \Omega(\eps D')$ vertices of $G$, which implies that
there are $O({\frac{n}{\eps D'}})$ centers.
\end{proof}

Recall that $m$ be the number of edges in the initial graph. 
We first show how to implement Step 2 of the initialization in time $O({\frac{nm}{\eps D'}})$.
More specifically, we show that the time spent to find a new node to add to $S$ or
to determine that there are  no further suitable nodes takes time $O(m)$. We test the
vertices of $V$ in some fixed order (say in order of increasing ID) whether they can be added to $S$. Assume the last
test that added a vertex to $S$ was the $t$-th test. At that time we set a counter at
every node in $G'$ to 0. Then we perform BFS (using only out-edges) from an arbitrary node $v$ of $G'$.
If $|Out(v, \eps' D'/2, G')| \ge
\eps' R'/2$ then we add $v$ to $S$, otherwise we store $|Out(v, \eps' D'/2, G')|$ as
counter at $v$
and start a new BFS. However, this and all later BFS traversals do not recursively
call themselves on nodes with non-zero counter: Let $y$ be a node at which a BFS started.
 If the BFS reaches a
node $x$ with non-zero counter, it does not explore the outcoming edges of $x$. Instead it uses the counter of $x$ to update its internal variable that stores the number of nodes
that can be reached from $y$ and continue the BFS by a backtracking step. When the BFS of $y$ has
terminated, we test $|Out(y, \eps' D'/2, G')|$ and proceed with $y$ as with $v$ before. If all nodes of $G'$ have been explored and none was added to $S$, 
the loop in Step 1 ends. Note that in this way every edge in $G'$ is traversed at most
once and then either a new node is added to $S$ or Step 2 terminates. 
Together with Claim~\ref{claim:centernum2} it follows that Step 2 takes time
 $O({\frac{nm}{\eps D'}})$. 

The other steps of the initialization can be implemented in time $O(m)$ each.
Thus the total time for initialization apart from the time
spent for initializing the in-$\mathcal{A}_{inc}$
data structures is $O({\frac{nm}{\eps D'}})$. 

The total time for initializing and maintaining the in-$\mathcal{A}_{inc}$
data structures is $O(T_{inc}(n,m,D',\eps){\frac{n}{\eps D'}} )$. 
Maintaining the heaps ${\cal H}(v)$ and the global heap $\mathcal H$  increases the running time only by a
factor of $O(\log n)$ as the number of heap operations is at most linear in 
$O(T_{inc}(n,m,D',\eps){\frac{n}{\eps D'}} )$.
To summarize, the total time for the algorithm is $O((T_{inc}(n,m,D',\eps) \log n + m) {\frac{n}{\eps D'}} )$. 
\end{proof}

\subsubsection{Radius}


\begin{theorem}\label{thm:incredrad2} 
There is a deterministic algorithm  for incremental radius in unweighted, directed  graphs that, for any   $\eps$ with  $0<\eps<2$, runs in total time $\tilde{O}(\max_{R_f\leq R'\leq R_0} \{(T_{inc}(n,m,2R',\eps) +
m) n/(\eps^2 R') \} + m n)$, 
and maintains an estimate $\hat{R}$ such that $R\leq \hat{R}\leq (1+\eps) R$, where $R$ is the radius of the current graph.
For strongly connected graphs, the additive factor of $mn$ in the time is not necessary.
\end{theorem}



By Lemma~\ref{lem:min}, the following lemma implies Theorem~\ref{thm:incredrad2}.

\begin{lemma}\label{lem:incredrad2}
There is a deterministic algorithm for incremental radius in unweighted, directed graphs that,  given a parameter $R'$, for any $\eps$ with $0<\eps<2$, runs in total time $\tilde{O}((T_{inc}(n,m,2R',\eps) + m) n/{(\eps R')}+mn)$, 
and maintains an estimate $\hat{R} \ge R$ such that if $R' \ge R$ then  $\hat{R} \le (1+\eps)R'$, where $R$ is the radius of the current graph. For strongly connected graphs, the additive factor of $mn$ in the time is not necessary.
\end{lemma}

\begin{proof}$  $

The basic idea for the proof of the lemma is that it is possible deterministically (1) to pick a set of {\em center} nodes, (2) to compute the eccentricity of each of the centers
and (3) to return the minimum values of these eccentricities. Let $s^*$ be a vertex with $\epsilon(s^*) = R$. We note that in this proof we do not refer to $s^*$ as the center --- instead a center is a node in the set chosen in step (1) above.  Note that $s^*$ has a path to
every node of length at most $R$. Let $I$ be the set of nodes
 that can reach $s^*$ by a path of length at most $\eps R$. Then  for every vertex
$x \in I$ it must hold that
$x$ has to every other node a path of length at most $(1+\eps)R$  and, thus,  $x$ has an eccentricity of at most $(1+\eps) R$. Hence, we will chose centers in such a way that
it is guaranteed that one of the nodes in $I$ is a center. In reality, the proof of correctness is slightly more complicated as the SSSP data structure that we use as subroutine only gives approximate distances.

Although our goal for choosing centers is essentially the same as that for diameter, we cannot directly apply the algorithm from diameter because step 3 assumes the graph is strongly connected. For diameter, we can assume that the graph is strongly connected because if it is not then the diameter is infinite. In contrast, the radius of a directed graph that is not strongly connected can be finite.

We can bypass this issue by using an algorithm for picking centers similar to that from the diameter algorithm, but only running this algorithm on the strongly connected component (SCC) that contains $s^*$. This SCC may grow so we also need to dynamically augment the set of centers.

\paragraph{Algorithm}
We run a deterministic incremental algorithm that maintains SCCs as a subroutine. An algorithm of Haeupler, Sen, and Tarjan~\cite{haeupler2008incremental} does this in total time $O(m^{3/2})$.

We let a \emph{top} SCC be the SCC that can reach all vertices in the graph. Note that there are either 0 or 1 top SCCs. Given the vertices in the top SCC (if it exists), we also maintain the graph $H$ induced by these vertices. (This is simple: whenever a vertex is added to the top SCC we iterate through its incident edges to see which are to or from a vertex in the top SCC, and after each edge insertion we check whether both of its endpoints are in the top SCC.) We note that all vertices not in $H$ have infinite eccentricity, so if $H$ is nonempty, then $H$ must contain the vertex $v$ for which $\varepsilon(v)=R$.

Let $\eps'=\eps/2$. When we run algorithm out-$\mathcal{A}_{inc}$ from a vertex $v$, we also keep a max-heap ${\cal H}(v)$ that stores for each other vertex $u$ the estimate $d'(v,u)$ returned by the algorithm and which outputs the value $\hat d_{out}(v)
=\max_u d'(v,u)/(1-\eps')$. Additionally we keep a min-heap $\mathcal{H}$ which stores for each center $v$ the value
$\hat  d_{out}(v)$.
This heap is updated whenever $\hat d_{out}(v)$
changes for any center $v$.


\subparagraph{Initialization}
We deterministically choose the centers in exactly the same way as for diameter, except $D'$ is replaced with $R'$.

\subparagraph{Edge insertion}
If at least one new vertex has been added to $H$ as a result of the edge insertion, we update the set of centers as follows. We simply rerun the center picking algorithm described in the running time analysis for diameter using the graph $H$ as input. That algorithm tests the vertices in increasing order of ID whether they can be added to $S$; before running it we change the IDs of the vertices newly added to $H$ so that they are tested after all of the old vertices. This way, the new set of centers is a superset of the old set of centers.

Then, for every center $c$, we run the algorithm out-$\mathcal{A}_{inc}(c,2R',\eps')$ 
and update the corresponding heaps.


\subparagraph{Queries}
The algorithm returns the minimum value in $\mathcal{H}$.

\paragraph{Analysis}

\subparagraph{Correctness}

By an identical proof to Claim~\ref{claim:near2} (with $D'$ replaced by $R'$), after every edge update there exists for every vertex $u$ in $H$  a center  $v$ such that  $d(v,u) \le \eps' R'$. Recall that $s^*$ is in $H$.

Now the correctness proof proceeds as follows:
Recall that we aim to show that (i) $\hat{R} \geq R$ and (ii) if $R'\geq R$, 
then $\hat{R} \le (1+\eps) R'.$

(i) For each center $v$  the value $\hat d_{out}(v) = \max_u d'(v,u)/(1-\eps') \ge \varepsilon(v)$ and the algorithm outputs
$\min_{v \in S} \hat d_{out}(v)$.
Then $R = \min_x \varepsilon(x) \le \min_{v \in S} \varepsilon(v) \le  \min_{v \in S} \hat d_{out}(v) / (1-\eps')= \hat{R}$.  

(ii)
Let $s^*$  be a vertex with $\varepsilon(s^*)= R$ and let $v^*$ be a center 
 with  $d( v^*,s^*) \le \eps' R'$. 
Note that $\varepsilon(s^*)= R$
 implies $\varepsilon(v^*) \le \varepsilon(s^*) + \eps' R' = R + \eps' R'.$
 If $R' \geq R$, it follows that $\varepsilon(v^*) \le (1 + \eps') R'$.
As the value $\hat{R}$ returned by a query is $\min_{v \in S} \hat d_{out}(v)/(1-\eps')$, it follows that
$\hat{R} \le \varepsilon(v^*)/(1-\eps') \le \frac{1+\eps'}{1-\eps'} R'$.
Setting $\eps'=\eps/2$ completes the correctness proof.

\subparagraph{Running time}
An identical argument to that of diameter (with $D'$ replaced by $R'$) shows that without the SCC algorithm and the algorithm for updating centers upon edge insertion, the running time is $\tilde{O}((T_{inc}(n,m,2R',\eps) + m) n/{(\eps R')})$. The SCC algorithm runs in time $O(m^{3/2})$. The algorithm for updating centers runs in time $O(m)$ from the running time analysis for diameter. New vertices are added to $H$ at most $n$ times, so the total time for updating centers is $O(mn)$. Thus, the total time is $\tilde{O}((T_{inc}(n,m,2R',\eps) + m) n/{(\eps R')}+mn)$.
\end{proof}

\subsubsection{Eccentricities}

\begin{theorem}\label{thm:detecc} 
There is a deterministic algorithm  for incremental eccentricities in unweighted, directed, strongly connected graphs that, for any   $\eps$ with  $0<\eps<2$, runs in total time \\$\tilde{O}(\max_{R_f\leq D'\leq D_0} \{(T_{inc}(n,m,D',\eps) +
m) n/(\eps^2 D') \})$, 
and for all $v\in V$ maintains an estimate $\hat{\varepsilon}(v)$ such that $(1-\eps)\varepsilon(v)\leq \hat{\varepsilon}(v)\leq \varepsilon(v)$ where $\varepsilon(v)$ is the eccentricity of $v$ in the current graph.
\end{theorem}


By Lemma~\ref{lem:max}, the following lemma implies Theorem~\ref{thm:detecc}.

\begin{lemma}
There is a deterministic algorithm for incremental eccentricities in unweighted, directed, strongly connected graphs that,  given a parameter $D'$, for any $\eps$ with $0<\eps<2$, runs in total time $\tilde{O}((T_{inc}(n,m,D',\eps) + m) n/{(\eps D')}+mn)$, 
and for all $v\in V$ maintains an estimate $\hat{\varepsilon}(v)\leq \varepsilon(v)$ such that if $D'\leq \varepsilon(v)$, then $\hat{\varepsilon}(v) \geq (1-\eps)\varepsilon(v)$.
\end{lemma}


\begin{proof}$  $

The basic idea for the proof of the lemma is that it is possible deterministically (1) to pick a set of {\em center} nodes, (2) to compute the distance from every vertex to each of the centers 
and (3) to return for each vertex $v$ the maximum distance from $v$ to a center. Fix $v$ and let $v'$ be the farthest vertex from $v$. Let $I$ be the set of nodes
 that can reach $v'$ by a path of length at most $\eps (\varepsilon(v))$. Then  for every vertex
$x \in I$ it must hold that $d(v,x)\geq (1-\eps)\varepsilon(v)$.
Hence, we will chose centers in such a way that
it is guaranteed that one of the nodes in $I$ is a center. In reality, the proof of correctness is slightly more complicated as the SSSP data structure that we use as subroutine only gives approximate distances.


\paragraph{Algorithm}

Let $\eps'=\eps/2$. For every vertex $v$ we keep a max-heap ${\cal H}(v)$ that stores for each center $u$ the estimate $d'(v,u)$ returned by the algorithm in-$\mathcal{A}_{inc}$.

\subparagraph{Initialization}
We deterministically choose the centers in exactly the same way as for diameter

\subparagraph{Edge insertion}

Then, for every center $c$, we run the algorithm in-$\mathcal{A}_{inc}(c,D',\eps')$ 
and update the corresponding heaps.

\subparagraph{Queries}

Given a query for $\varepsilon(v)$, the algorithm returns the maximum value in $\mathcal{H}(v)$.

\paragraph{Analysis}

\subparagraph{Correctness}

By an identical proof to Claim~\ref{claim:near2}, after every edge update there exists for every vertex $u$ a center  $v$ such that  $d(v,u) \le \eps' D'$. Fix $v$ and let $v'$ be such that $d(v,v')=\varepsilon(v)$

Now the correctness proof proceeds as follows:
Recall that we aim to show that (i) $\hat{\varepsilon}(v)\leq \varepsilon(v)$ and (ii) if $D'\leq \varepsilon(v)$, then $\hat{\varepsilon}(v) \geq (1-\eps)\varepsilon(v)$.

(i) The return value $\hat d_{out}(v)$ is for some $u$, $d'(v,u)\leq d(v,u)\leq\varepsilon(v)$.

(ii)
Let $v^*$ be a center 
 with  $d( v^*,v') \le \eps' D'$. So, if $D'\leq \varepsilon(v)$ then $d( v^*,v') \le \eps'(\varepsilon(v))$. Then by the triangle inequality, $d(v,v^*)\geq (1-\eps')\varepsilon(v)$. Thus, $d'(v,v^*)\geq (1-\eps')^2\varepsilon(v)$. As the value $\hat{\varepsilon}(v)$ returned by a query is $\max_{u \in S} d'(v,u)$, it follows that $\hat{\varepsilon}(v)\geq (1-\eps')^2\varepsilon(v)$. 
Setting $\eps'=\eps/2$ completes the correctness proof.

\subparagraph{Running time}
An identical argument to that of diameter shows that the running time is\\ $\tilde{O}((T_{inc}(n,m,D',\eps) + m) n/{(\eps D')})$.
\end{proof}